\documentclass[10pt]{article}
\usepackage{amsfonts,amsmath,amssymb,amsthm}
\usepackage[pdftex]{graphicx}
\usepackage[hmargin=1in,vmargin=1in]{geometry}
\usepackage{natbib} 
\usepackage{setspace}
\usepackage{enumerate}
\usepackage[toc,title,titletoc,header]{appendix}
\usepackage[colorlinks,citecolor=blue,hypertexnames=false]{hyperref}
\usepackage{etoolbox}
\usepackage{booktabs}
\usepackage{multirow}
\allowdisplaybreaks

\def\qed{\rule{2mm}{2mm}}

\newtheorem{theorem}{Theorem}[section]
\newtheorem{lemma}{Lemma}[section] 

\newtheorem{corollary}{Corollary}[section]

\theoremstyle{definition}

\newtheorem{remark}{Remark}[section]
\newtheorem{assumption}{Assumption}[section]

\AtEndEnvironment{remark}{~\qed}
\AtEndEnvironment{example}{~\qed}

\DeclareMathOperator{\var}{Var}
\DeclareMathOperator{\cov}{Cov}
\DeclareMathOperator{\diag}{diag}
\DeclareMathOperator{\tr}{tr}
\DeclareMathOperator{\blp}{BLP}

\begin{document}

\title{\vspace{-.5in}A New Design-Based Variance Estimator for Finely Stratified Experiments \thanks{We thank Xinran Li and seminar participants at the University of Southern California for helpful comments. We thank Juri Trifonov for excellent research assistance.  The fourth author acknowledges support from the National Science Foundation through grant SES-2419008. The fifth author acknowledges support from the National Science Foundation through the grant SES-2149408.}}

\author{Yuehao Bai \\
Department of Economics \\
University of Southern California \\
\url{yuehao.bai@usc.edu}
\and
Xun Huang \\
Department of Economics \\
University of Chicago \\
\url{xhuang520@uchicago.edu}
\and
Joseph P.\ Romano\\
Departments of Economics \& Statistics \\
Stanford University\\
\url{romano@stanford.edu}
\and
Azeem M.\ Shaikh \\
Department of Economics \\
University of Chicago \\
\url{amshaikh@uchicago.edu}
\and
Max Tabord-Meehan \\
Department of Economics \\
University of Chicago \\
\url{maxtm@uchicago.edu}
}

\begin{spacing}{1.2}
\maketitle
\end{spacing}

\vspace{-0.2in}

\begin{spacing}{1}
\begin{abstract}
This paper considers the problem of design-based inference for the average treatment effect in finely stratified experiments.  Here, by ``design-based'' we mean that the only source of uncertainty stems from the randomness in treatment assignment itself; by ``finely stratified'' we mean that units are first stratified into groups of a fixed size according to baseline covariates and then, within each group, a fixed number of units are assigned uniformly at random to treatment and the remainder to control. In this setting we present a novel estimator of the variance of the difference-in-means based on pairing ``adjacent" strata. Importantly, our estimator is well defined even in the challenging setting where there is exactly one treated or control unit per stratum. We prove that our estimator is upward-biased, and thus can be used for inference under mild restrictions on the finite population. We compare our estimator with some well-known estimators that have been proposed previously in this setting, and demonstrate that, while these estimators are also upward-biased, our estimator has smaller bias and therefore leads to more precise inferences whenever ``adjacent" strata are sufficiently similar. To further understand when our estimator leads to more precise inferences, we then introduce a framework motivated by a thought experiment in which the finite population can be modeled as having been drawn once in an i.i.d.\ fashion from a well-behaved probability distribution. In this framework, we argue that our estimator dominates the others in terms of limiting bias and that these improvements are strict except under exceptionally strong restrictions on the treatment effects.  Finally, we illustrate the practical relevance of our theoretical results through a simulation study, which reveals that our estimator can in fact lead to substantially more precise inferences, especially when the quality of stratification is high.
\end{abstract}
\end{spacing}



\thispagestyle{empty} 
\newpage
\setcounter{page}{1}

\section{Introduction} \label{sec:intro}

This paper considers the problem of design-based inference on the average treatment effect in finely stratified experiments.  Here, by ``design-based'' we mean that there is no sampling uncertainty and the only randomness stems from the variation in treatment assignment itself; by ``finely stratified'' we mean that treatment is assigned in the following fashion: units are first stratified into groups of a fixed size $k$ according to baseline covariates and then, within each group, a fixed number $\ell < k$ are assigned uniformly at random to treatment and the remainder to control.   A prominent special case of this framework is a matched pairs design, in which $k = 2$ and $\ell = 1$. While our primary focus is on experiments where treatment is assigned at the level of the individual, in Remark \ref{rem:cluster} we explain how our analysis can accommodate experiments where the treatment is assigned at the cluster level.

In this setting, we first establish, by way of motivation, a result that shows under mild conditions that inference using the usual difference-in-means estimator requires an estimator of its variance that is at least asymptotically upward-biased. After having motivated the importance of such estimators, we review the standard argument for why constructing such estimators is particularly challenging for the case in which $\ell = 1$ or $k - \ell = 1$; in particular, in this case the sample variance of the outcomes under treatment (if $\ell = 1$) or control (if $k - \ell = 1$) computed in each stratum are identically zero. With this in mind, we propose a novel estimator that remains well-defined even in this challenging case, and demonstrate that our estimator is upward-biased under minimal assumptions. The construction of our variance estimator involves a pairing of the strata, and its bias depends on the corresponding differences of the average treatment effects in these paired strata. As a result, if the strata are paired in such a way that ensures ``similar'' strata are adjacent, then the bias of our estimator will naturally be small.

We then compare our variance estimator with a canonical variance estimator that has been proposed previously for the case in which $\ell = 1$ or $k - \ell = 1$. The estimator, proposed by \cite{imai2008variance}, is based on the sample variance obtained by viewing the difference-in-means estimates for each stratum as independent observations.  Existing results demonstrate that this estimator is also upward-biased for the variance of the difference-in-means estimator, but we demonstrate that the bias of our estimator is lower (and as a consequence leads to more precise inferences) whenever we can ensure that the ``adjacent" strata used to construct our estimator are sufficiently similar. To further distinguish our estimator from the estimator proposed in \cite{imai2008variance}, as well as a second estimator proposed by \cite{fogarty2018mitigating} which modifies the Imai estimator by partialling out the average covariate values in each stratum, we introduce a framework motivated by a thought experiment in which the finite population can be modeled as being drawn once in an i.i.d.\ fashion from a well-behaved probability distribution. %Importantly, here we maintain that the strata are formed in such a way that the covariate values of the units in each stratum, as well as the paired strata, get close in the limit. 
In this framework, we show that the limiting bias of our estimator is strictly smaller than that of \cite{imai2008variance} unless the treatment effects are homogeneous, and strictly smaller than that of \cite{fogarty2018mitigating} unless the conditional average treatment effect is linear in the covariates. As a result, confidence intervals based on our variance estimator will be strictly shorter than those based on the other variance estimators unless these exceptional restrictions hold.

The literature on stratified block randomization dates back to at least \cite{fisher1935design}. For the case in which $\ell = 1$ or $k - \ell = 1$, the variance estimator in \cite{imai2008variance} has been extended to different settings by \cite{imai2009essential}, \cite{pashley2021insights}, and \cite{zhu2024design-based}. \cite{fogarty2018regression-assisted} study regression adjustment for matched pair designs, and \cite{liu2020regression-adjusted} studies regression adjustment for stratified designs primarily for the case where $\min\{\ell, k - \ell\} > 1$. \cite{ding2017paradox} considers randomization inference for matched pairs, alongside other randomization schemes.  We note that all of these papers employ a design-based framework as we do in this paper. As a consequence, their analyses differ from work which studies inference in finely stratified experiments from a superpopulation perspective \citep[][]{bai2022optimality, bai2022inference, bai2024inference-b, bai2024covariate, bai2024inference, jiang2024bootstrap, cytrynbaum2024covariate, cytrynbaum2021optimal, bai2025efficiency}; see also \cite{abadie2008estimation}, who consider superpopulation estimation of the variance in matched pair designs, conditional on the covariates, in an alternative sampling framework where pairs are sampled instead of the units.

The remainder of the paper is organized as follows. In Section \ref{sec:setup}, we describe our setup and notation.  Section \ref{sec:main} contains our main results. There we first highlight the importance of upward-biased variance estimators, and then review existing proposals before introducing our novel variance estimator. In Section \ref{sec:variances}, we compare our variance estimators with existing ones in a particular limiting thought experiment. We examine the finite-sample behavior of confidence intervals based on all these estimators through a simulation study in Section \ref{sec:sims}.

\section{Setup and Notation} \label{sec:setup}

Consider an experiment consisting of $i \in \{1, \ldots, n\}$ units. For the $i$th unit, let $Y_{i} \in \mathbf{R}$ denote their observed outcome, $D_i \in \{0, 1 \}$ denote their received treatment, $X_i \in \mathbf{R}^{p}$ denote their observed, baseline covariates, and $Y_i(d)$ denote their potential outcome under treatment $d \in \{0,1\}$. As usual, the observed outcomes are related to the potential outcomes via the relationship
\begin{equation}\label{eq:obsY}
Y_i = Y_i(1)D_i + Y_i(0)(1 - D_i)~.
\end{equation}
In what follows, it will be convenient to use the following shorthand notation: for a generic random vector $A_i$ indexed by $i$, let $A^{(n)} := (A_i: 1 \leq i \leq n)$.  For later use, we also define $W_i := (Y_i(1),Y_i(0),X_i)'$. 

In the design-based framework that we maintain throughout the paper, the potential outcomes and covariates of the units in the experiment are modeled as nonrandom quantities, with the only source of randomness arising from the treatment assignment mechanism. Our analysis concerns ``finely stratified'' designs in which the covariates are used to stratify units in the experiment into groups, i.e., strata, of \emph{fixed} size $k$, and then, within each group, $\ell < k$ units are chosen uniformly at random to assign to treatment and the remainder to control; of particular interest to us will be the special case when $\min\{\ell, k - \ell\} = 1$. For ease of exposition, we assume throughout that $n = mk$, so that $m$ denotes the number of strata. Formally, we model the strata as a partition of $\{1, \dots, n\}$:
$$\Lambda_n := \{ \lambda_j \subset \{1, \dots, n\}: 1 \leq j \leq m \}~,$$ 
with $|\lambda_j| = k$. It is worth emphasizing that we have suppressed in the notation the fact that each $\lambda_j$ (and therefore $\Lambda_n$ itself) can depend on $X^{(n)}$. Since we maintain that $k$ is fixed throughout the paper, when we write that $n \to \infty$, it should be understood that $m \to \infty$. Using this notation, the (joint) distribution of treatment assignment is characterized by the following assumption:

\begin{assumption} \label{ass:dn}
Treatment status $D^{(n)}$ is assigned independently for each $1 \leq j \leq m$ so that $$(D_i : i \in \lambda_j) \sim \text{Unif}\bigg (\bigg \{(d_1, \ldots, d_k) \in \{0,1\}^k : \sum_{1 \leq i \leq k} d_i = \ell \bigg \} \bigg)~.$$
\end{assumption}

\noindent In other words, $\ell$ out of $k$ units in each stratum are treated uniformly at random, independently across strata.

Our parameter of interest is the finite population average treatment effect, given by $$\Delta_n := \bar Y_n(1) - \bar Y_n(0)~,$$ where, for $d \in \{0,1\}$, 
$$\bar Y_n(d) := \frac{1}{n} \sum_{1 \leq i \leq n} Y_i(d)~.$$  
A natural estimator of $\Delta_n$ is given by the usual difference-in-means estimator, i.e., 
$$\hat \Delta_n := \frac{1}{n(1)} \sum_{1 \leq i \leq n} Y_i D_i - \frac{1}{n(0)} \sum_{1 \leq i \leq n} Y_i (1 - D_i)~,$$ 
where $n(1) = \sum_{1 \le i \le n} D_i = n \eta$, with $\eta := \ell/k$ the proportion of the $n$ units that are treated, and $n(0) = \sum_{1 \le i \le n}(1 - D_i) = n(1 - \eta)$.  Note that because we have assumed that $n = mk$, both $n \eta = m \ell$ and $n ( 1 - \eta) = m (k - \ell)$ are integers.

\begin{remark}\label{rem:cluster}
The analysis in this paper can be readily extended to cluster RCTs where the treatment is implemented at the cluster level \citep[see, for instance,][]{imai2009essential, su2021model, de2024level, bai2024inferencecluster}. Let $i \in \{1, \dots, n\}$ denote the $i$th cluster and $g_i$ denote the number of units in the $i$th cluster. Let $D_i$ denote the treatment assignment for the $i$th cluster, which we emphasize applies to all units in this cluster. For each $i$ and $1 \leq t \leq g_i$, let $Y_{i, t}(1)$, $Y_{i, t}(0)$ denote the nonrandom potential outcomes for the $t$th unit in the $i$th cluster and let $Y_{i, t}$ denote its observed outcome. Suppose the parameter of interest is the average treatment effect across all units:
\[ \Delta_n^\dagger = \frac{1}{\sum_{1 \leq i \leq n} g_i} \sum_{1 \leq i \leq n} \sum_{1 \leq t \leq g_i} (Y_{i,t}(1) - Y_{i,t}(0))~. \]
Note that $\Delta_n^\dagger = \frac{1}{n}\sum_{1 \leq i \leq n} (Y_i^\dagger(1) -  Y_i^\dagger(0))$, where for $d \in \{0, 1\}$,
\[ Y_i^\dagger(d) = \frac{1}{\frac{1}{n}\sum_{1 \leq i \leq n} g_i} \sum_{1 \leq t \leq g_i} Y_{i,t}(d)~. \]
Accordingly, let $Y_i^\dagger = Y_i^\dagger(1) D_i + Y_i^\dagger(0) (1 - D_i)$, and a natural estimator for $\Delta^{\dagger}_n$ is given by
\[ \hat \Delta_n^\dagger = \frac{1}{n(1)} \sum_{1 \leq i \leq n} Y_i^\dagger D_i - \frac{1}{n(0)}\sum_{1 \le i \le n} Y_i^\dagger (1 - D_i)~, \]
where $n(1)$ and $n(0)$ are defined as before, and in this case are the numbers of treated and untreated \emph{clusters}. The analysis in this paper then immediately applies by replacing $Y_i$ by $Y^{\dagger}_i$ throughout.
\end{remark}

\section{Main Results}\label{sec:main}
In this section, we begin, by way of motivation, with a result that highlights the importance of variance estimators that are at least \emph{asymptotically} upward-biased, defined precisely in the statement of  Theorem \ref{thm:upward} below.  Under mild regularity conditions, we show, in particular, that valid inference for $\Delta_n$ based on a standard $t$-statistic requires a variance estimator that is asymptotically upward-biased for $\var[\hat{\Delta}_n]$.  Of course, a simple, sufficient condition for an estimator to be asymptotically upward-biased is that it is upward-biased.  In Section \ref{sec:review}, we review some prior proposals for constructing variance estimators in finely stratified experiments that are upward-biased under only the assumption that $D^{(n)}$ satisfies Assumption \ref{ass:dn}.  We further discuss in Remark \ref{rem:better_bounds} some improvements upon these estimators that typically result in estimators that are only asymptotically upward-biased under assumptions stronger than Assumption \ref{ass:dn} and therefore result in valid inference less generally than their upward-biased counterparts.  Finally, in Section \ref{sec:main_var}, we present our main results.  We introduce a novel estimator of $\var[\hat{\Delta}_n]$ for finely stratified experiments and show that, like the estimators that we review in Section \ref{sec:review}, it is upward-biased for $\var[\hat{\Delta}_n]$ under only the assumption that $D^{(n)}$ satisfies Assumption \ref{ass:dn}. In contrast to these other estimators, however, we further argue that the magnitude of the bias of our estimator depends explicitly on the quality of the groupings of the experimental units into strata, in a sense made precise by Theorem \ref{thm:main} below. 

\subsection{Motivating Result}\label{sec:motivation}
The following theorem formalizes the sense in which an asymptotically upward-biased estimator for $\var[\hat{\Delta}_n]$ is required for valid inference on $\Delta_n$. 

\begin{theorem}\label{thm:upward}
Suppose $D^{(n)}$ satisfies Assumption \ref{ass:dn} and consider a sequence of finite populations such that 
\begin{equation}\label{eq:pop_moments}
\frac{1}{n}\max_{d \in \{0, 1\}}\max_{1 \leq i \leq n}Y_i(d)^2 \rightarrow 0
\end{equation}
as $n \rightarrow \infty$, and 
\begin{equation}\label{eq:pop_nondegen}
0 < \liminf_{n \to \infty} n \cdot \var[\hat{\Delta}_n] \le \limsup_{n \to \infty} n \cdot \var[\hat \Delta_n] < \infty~.
\end{equation}
Further assume that $\tilde{V}_n$ is an estimator of $\var[\hat \Delta_n]$ such that
\begin{equation} \label{eq:consistent}
n\cdot \big | \tilde{V}_n - E[\tilde{V}_n] \big | \xrightarrow{P} 0
\end{equation}
as $n \rightarrow \infty$. Then,
\begin{equation} \label{eq:valid}
\liminf_{n \rightarrow \infty} P\Big\{\hat{\Delta}_n - \sqrt{\tilde{V}_n}\cdot z_{1-\alpha/2} \le \Delta_n \leq \hat{\Delta}_n + \sqrt{\tilde{V}_n}\cdot z_{1 - \alpha/2}\Big\} \ge 1 - \alpha~,    
\end{equation}
where $z_{1-\alpha/2}$ is the $1 - \alpha/2$ quantile of the standard normal distribution, if and only if
\begin{equation} \label{eq:upward}
\liminf_{n \to \infty} n \big ( E[\tilde{V}_n] - \var[\hat{\Delta}_n] \big ) \geq 0~.
\end{equation}
\end{theorem}

Note that because $\tilde{V}_n$ is defined as an estimator of $\var[\hat \Delta_n]$ and not the asymptotic variance of $\sqrt{n}(\hat \Delta_n - \Delta_n)$, the conditions in the theorem are all stated with a scaling by $n$. When \eqref{eq:upward} holds, we say that $\tilde{V}_n$ is asymptotically upward-biased for $\var[\hat{\Delta}_n]$.  A simple sufficient condition for \eqref{eq:upward} is, of course, that 
\begin{equation} \label{eq:exactupward}
E[\tilde{V}_n] \ge \var[\hat \Delta_n] \text{ for all } n \geq 1~.
\end{equation}
When \eqref{eq:exactupward} holds, we say that $\tilde{V}_n$ is upward-biased for $\var[\hat{\Delta}_n]$. Theorem \ref{thm:upward} thus demonstrates that, under mild regularity conditions, confidence intervals for $\Delta_n$ of the form $[\hat \Delta_n \pm  \sqrt{\tilde{V}_n}\cdot z_{1 - \alpha/2}]$ are valid in the sense of having the right limiting coverage probability,  if and only if $\tilde{V}_n$ is asymptotically upward-biased for $\var[\hat{\Delta}_n]$. These assumptions include weak restrictions on the finite population (i.e., \eqref{eq:pop_moments}), a non-degeneracy condition (i.e., \eqref{eq:pop_nondegen}), and a requirement that $n\cdot\tilde{V}_n$ is consistent for its expectation (i.e., \eqref{eq:consistent}). This last condition will typically hold under some further restrictions on finite population moments; see, e.g., Theorem \ref{thm:main_var} below. Finally, we note that it is immediate from this result that the length of the interval $[\hat \Delta_n \pm  \sqrt{\tilde{V}_n}\cdot z_{1 - \alpha/2}]$ will necessarily be (asymptotically, weakly) shorter whenever it is constructed using a variance estimator $\tilde{V}_n$ with correspondingly smaller asymptotic upwards bias. As a consequence, variance estimators with smaller bias lead to (asymptotically) more precise inferences. 

The estimators discussed in Sections \ref{sec:review} and \ref{sec:main_var} below can be shown to satisfy \eqref{eq:exactupward}, and therefore \eqref{eq:upward}, whenever $D^{(n)}$ satisfies Assumption \ref{ass:dn}.  Such estimators therefore lead to confidence intervals for $\Delta_n$ that are valid whenever the hypotheses of Theorem \ref{thm:upward} are satisfied.  It is worth emphasizing, however, that establishing \eqref{eq:upward} for other estimators may in general require additional assumptions beyond the hypotheses of Theorem \ref{thm:upward}.  In Remarks \ref{rem:better_bounds} and \ref{rem:collapse} we discuss some examples of such estimators, and we illustrate using our simulations in Section \ref{sec:sims} that their resulting confidence intervals may lead to under-coverage when these additional conditions are not satisfied. 

\subsection{Review of Some Upwards Biased Estimators of $\var[\hat{\Delta}_n]$}\label{sec:review}

In this section, we review some upward-biased estimators of $\var[\hat{\Delta}_n]$.  To this end, recall that it can be shown using standard arguments \citep[see, for instance,][]{imbens2015causal} that
\begin{equation} \label{eq:varDeltahat}
\var[\hat{\Delta}_n] = \frac{1}{nm}\sum_{1 \leq j \leq m}\left(\frac{S^2_j(1)}{\eta} + \frac{S^2_j(0)}{1 - \eta} - S_{j, \Delta}^2\right)~,    
\end{equation}
where
\[S^2_j(d) := \frac{1}{k-1}\sum_{i \in \lambda_j}\left(Y_i(d)- \bar Y_{j, n}(d)\right)^2~, \hspace{5mm} S_{j, \Delta}^2 := \frac{1}{k-1}\sum_{i \in \lambda_j}\left(Y_i(1) - Y_i(0) - \Delta_{j,n}\right)^2~,\]
with
\[\bar Y_{j, n}(d) := \frac{1}{k}\sum_{i \in \lambda_j}Y_i(d)~, \hspace{5mm} \Delta_{j,n} := \frac{1}{k}\sum_{i \in \lambda_j}\left(Y_i(1) - Y_i(0)\right)~.\]
In settings where $\min\{\ell, k - \ell\} > 1$, the construction of an upward-biased variance estimator is straightforward: an unbiased estimator of $S^2_j(d)$ for $d \in \{0, 1\}$ is given by the sample variance of the outcomes for units within stratum $j$ assigned to treatment $d$, which we denote by $\hat{S}^2_j(d)$ \citep[see][]{imbens2015causal, pashley2021insights}. A simple estimator of $\var[\hat{\Delta}_n]$ is thus given by
\begin{equation}\label{eq:var_coarse}\frac{1}{nm}\sum_{1 \le j \le m}\left(\frac{\hat{S}^2_j(1)}{\eta} + \frac{\hat{S}^2_j(0)}{1 - \eta}\right)~,
\end{equation}
and its bias is $\frac{1}{nm}\sum_{1 \le j \le m}S_{j, \Delta}^2 \ge 0$. Note that this estimation strategy exploits the lower bound $S_{j, \Delta}^2 \ge 0$; we return to this observation in Remark \ref{rem:better_bounds} below.

This estimation strategy fails, however, when there is exactly one treated or control unit in a stratum, i.e., when $\min\{\ell, k - \ell\} = 1$, since in this case the corresponding sample variance is identically zero.  In these settings, a canonical estimator considered in the literature instead computes the variance of the stratum-level average treatment effect estimates \citep[][]{imai2008variance}:
\[\hat{V}_n^{\rm IM} := \frac{1}{m (m - 1)}\sum_{1 \leq j \leq m}\left(\hat{\Delta}_{j,n} - \hat{\Delta}_n\right)^2~,\]
where
\[\hat{\Delta}_{j,n} := \frac{1}{\ell}\sum_{i \in \lambda_j}Y_iD_i - \frac{1}{k - \ell}\sum_{i \in \lambda_j}Y_i(1-D_i) \]
is the difference in means in the $j$th stratum. This variance estimator serves as the basic scaffolding for many of the recent estimators proposed in the literature on inference in stratified randomized experiments: it is a special case of the small-block variance estimator proposed in \cite{pashley2021insights} when all strata are the same size (see also \cite{zhu2024design-based}), a special case of the pair-cluster variance estimator of \cite{de2024level} in the setting of an individual-level randomized experiment, and a special case of the regression-based estimator due to \cite{fogarty2018mitigating}, which we introduce in Section \ref{sec:variances}. Note that the bias of $\hat{V}^{\rm IM}_n$ is given by \citep[see][]{imbens2015causal,fogarty2018mitigating}
\begin{equation} \label{eq:imai-bias}
E[\hat{V}^{\rm IM}_n] - \var[\hat{\Delta}_n] = \frac{1}{m(m-1)}\sum_{1 \leq j \leq m}(\Delta_{j,n} - \Delta_n)^2 \geq 0~,    
\end{equation}
so that $\hat{V}^{\rm IM}_n$ is an upward-biased estimator of $\var[\hat{\Delta}_n]$. Moreover, under mild assumptions, it is straightforward to establish using Chebyshev's inequality that $\hat V_n^{\rm IM}$ is consistent for its expectation in the sense of \eqref{eq:consistent}, and thus can be used to construct valid confidence intervals in the sense of \eqref{eq:valid}.

\begin{remark}\label{rem:better_bounds}
In some circumstances, it may be further possible to bound $S_{j, \Delta}^2$ from below by a positive number instead of zero, and this tighter bound could then be used to construct less conservative estimators of $\var[\hat \Delta_n]$. For instance, it follows from the Cauchy-Schwarz inequality that
\[ \sum_{i \in \lambda_j} (Y_i(1) - \bar Y_{j, n}(1)) (Y_i(0) - \bar Y_{j, n}(0)) \leq \bigg ( \sum_{i \in \lambda_j} (Y_i(1) - \bar Y_{j, n}(1))^2 \bigg )^{1/2} \bigg ( \sum_{i \in \lambda_j} (Y_i(0) - \bar Y_{j, n}(0))^2 \bigg )^{1/2}~,\]
from which we can deduce the lower bound 
\[S_{j, \Delta}^2 \ge (S_j(1) - S_j(0))^2 \ge 0~.\]
In fact, it is sometimes possible to achieve even tighter bounds; see \cite{aronow2014sharp} for details. We note, however, that we are not aware of estimators based on lower bounds other than $S_{j, \Delta}^2 \ge 0$ which are necessarily guaranteed to be upward-biased, and in general additional assumptions  need to be imposed in order to guarantee that the resulting estimators are at least asymptotically upward-biased; see, for instance, Corollary 1 in \cite{aronow2014sharp}. As a consequence, these improved estimators may not lead to valid confidence intervals whenever these additional assumptions fail to hold. 
\end{remark}

\subsection{A Novel Paired-Strata Estimator of $\var[\hat{\Delta}_n]$}\label{sec:main_var}
Building on our earlier work on super-population approaches to inference for finely stratified designs \citep[][]{bai2022inference,bai2024inference-b,  bai2024covariate,bai2024inference,bai2025efficiency}, we now propose a novel estimator of $\var[\hat{\Delta}_n]$, primarily for settings where $\min\{\ell, k - \ell\} = 1$. Our estimator is constructed by pairing together strata; for simplicity, we pair adjacent strata together, so the pairs are given by $\{(\lambda_{2j - 1}, \lambda_{2j}): 1 \le j \le \lfloor m/2 \rfloor\}$, but we emphasize that this pairing should be understood as being obtained after permuting the indices of the strata $\{\lambda_j: 1 \leq j \leq m\}$.  In practice, it will be desirable to do this in a way that ensures ``similar'' strata are adjacent. With this is mind, our estimator is defined as follows:
\begin{align}\label{def::hat_V_n_mp}
\hat{V}_n := \frac{1}{m}\left(\hat{\tau}_n^2-\hat{\kappa}_n\right),
\end{align}
where
\[ \hat \tau_n^2 = \frac{1}{m} \sum_{1 \leq j \leq m} \hat \Delta_{j, n}^2\]
and
\[ \hat \kappa_n = \frac{2}{m} \sum_{1 \leq j \leq \lfloor \frac{m}{2} \rfloor} \hat \Delta_{2j - 1, n} \hat \Delta_{2j, n}~. \]
Note that it follows from straightforward algebraic manipulation that $\hat{V}_n \ge 0$. The construction of $\hat V_n$ can be motivated using a specific limiting thought experiment that we introduce in Section \ref{sec:variances} below.  The following theorem shows that $\hat V_n$ is an upward-biased estimator of $\var[\hat{\Delta}_n]$ when $D^{(n)}$ satisfies Assumption \ref{ass:dn} and, under mild additional restrictions, is consistent for its expectation.
\begin{theorem} \label{thm:main}
Suppose $D^{(n)}$ satisfies Assumption \ref{ass:dn}. Then,
\begin{enumerate}[\rm (a)]
\item The bias of $\hat{V}_n$ is given by
\[E[\hat{V}_n] - \var[\hat{\Delta}_n] = \xi^2_n := \frac{1}{m^2} \sum_{1 \leq j \leq \lfloor \frac{m}{2} \rfloor} (\Delta_{2j - 1, n} - \Delta_{2j, n})^2 \ge 0~. \]
\item Suppose in addition that 
\[\frac{1}{n}\sum_{1 \leq i \leq n}Y_i(d)^4 = o(n)~,\]
holds for $d \in \{0, 1\}$. Then, \eqref{eq:consistent} holds as $n \rightarrow \infty$. 
\end{enumerate}
Therefore, by combining Theorems \ref{thm:upward} and \ref{thm:main}, we immediately obtain the asymptotic validity of confidence intervals constructed using $\hat{V}_n$. 
\end{theorem}
Theorem \ref{thm:main} justifies the use of $\hat V_n$ for inference about $\Delta_n$ and further demonstrates that the magnitude of its bias depends on the differences of the stratum-level average treatment effects across pairs of strata, and hence the bias decreases whenever the pairs of strata are formed in a way that increases their homogeneity in this sense. The bias will therefore be small not only in the extreme instance in which the stratum-level average treatment effects are similar across all strata, but also when stratum-level average treatment effects vary across strata and are only similar for adjacent strata. Corollary \ref{cor:bias_sign} documents the relative biases of $\hat{V}_n$ versus $\hat{V}^{\rm IM}_n$. 

%The latter scenario is what one might naturally expect when strata are formed by grouping individuals with similar values of $X_i$. In contrast, the bias of $\hat{V}^{\rm IM}_n$ does not share this property. 
\begin{corollary}\label{cor:bias_sign}
Suppose $D^{(n)}$ satisfies Assumption \ref{ass:dn}. Then 
\[E[\hat{V}_n] - \var[\hat{\Delta}_n] \le E[\hat{V}^{\rm IM}_n] - \var[\hat{\Delta}_n]\]
if an only if 
\begin{equation}\label{eq:bias} \frac{1}{m} \sum_{1 \leq j \leq \lfloor \frac{m}{2} \rfloor} (\Delta_{2j - 1, n} - \Delta_n) (\Delta_{2j, n} - \Delta_n) \geq - \frac{1}{2m(m - 1)} \sum_{1 \leq j \leq m} (\Delta_{j, n} - \Delta_n)^2~. 
\end{equation}
\end{corollary}

We note that, a sufficient condition for \eqref{eq:bias} to hold is that the left-hand side is non-negative: we expect this to be the case whenever ``similar" strata are paired together. We argue that such a  condition holds quite generally for large $m$ in Section \ref{sec:variances} through additional restrictions on $W^{(n)}$ and $\Lambda_n$ that are motivated by a particular limiting thought experiment. 

\begin{remark}\label{rem:collapse}
An alternative construction for an estimator of $\var[\hat{\Delta}_n]$ based on our paired strata $\{(\lambda_{2j-1}, \lambda_{2j}):1 \le j \le \lfloor m/2 \rfloor\}$ would be to apply the estimation strategy described in \eqref{eq:var_coarse} to the paired strata directly. However, such an estimator is not guaranteed to be upward-biased in general. In particular, it can be shown that in the case of a matched pairs design (i.e., when $\ell = 1, k = 2$) that this estimator coincides with the estimator $\hat{V}^{\rm alt}_n$ proposed in Section \ref{sec:sims}. There we demonstrate via simulation that confidence intervals for $\Delta_n$ constructed using this estimator may not have correct coverage in general.
\end{remark}

\section{Comparison with Other Variance Estimators} \label{sec:variances}

In this section, we provide a framework that permits us to further discriminate among the variance estimators of $\var[\hat{\Delta}_n]$ discussed in Section \ref{sec:main}.  This framework is formalized by restrictions on the sequence of finite populations and experimental designs $(W^{(n)},\Lambda_n)$ specified in the assumptions below.  These restrictions are motivated by a limiting thought experiment in which $W^{(n)}$ is itself a realization of an i.i.d.\ sample from a fixed probability distribution, and the strata $\Lambda_n$ are formed in such a way that units with similar observable characteristics are grouped together, but, as explained in Remark \ref{rem:equi} below, the restrictions can also be shown to hold in other contexts as well.

\begin{assumption} \label{ass:super}
For some random variable $\tilde W := (\tilde Y(1), \tilde Y(0), \tilde X) \sim Q$ with $\max_{d \in \{0,1\}} \allowbreak E_Q[\tilde Y^2(d)] < \infty$, $W^{(n)}$ and $\Lambda_n$ satisfy the following requirements:
\vspace{-3mm}
\begin{enumerate}
\item[(a)] As $n \to \infty$,
\begin{align*}
    \frac{1}{n} \sum_{1 \leq i \leq n} Y_i^r(d) & \to E_Q[\tilde Y^r(d)] \text{ for } d \in \{0, 1\} \text{ and } r \in \{1, 2\} \\
    \frac{1}{n} \sum_{1 \leq i \leq n} Y_i(1) Y_i(0) & \to E_Q[\tilde Y(1) \tilde Y(0)]~.
\end{align*}
\item[(b)] For any $(d,d') \in \{0,1\}^2$ and sequence $\{(i_j,i_j'): (i_j,i_j') \in (\lambda_{2j - 1}\cup \lambda_{2j})^2, i_j \ne i_j', 1 \le j \le \lfloor m/2 \rfloor\}$, as $n \to \infty$, 
$$\frac{2}{m} \sum_{1 \leq j \leq \lfloor \frac{m}{2} \rfloor} Y_{i_j}(d) Y_{i_j'}(d') \stackrel{}{\rightarrow} E_Q[E_Q[\tilde Y(d) | \tilde X] E_Q[\tilde Y(d') | \tilde X]]~.$$
\end{enumerate}
\end{assumption}
For the analysis of the estimator proposed by \cite{fogarty2018mitigating}, introduced below, we need to further impose the following assumption:
\begin{assumption} \label{ass:X}
For $\tilde W \sim Q$ in Assumption \ref{ass:super}, $W^{(n)}$ and $\Lambda_n$ satisfy the following requirements:
\begin{enumerate}
\item[(a)] As $n \to \infty$,
\[ \frac{1}{n}\max_{1\le i \le n} \lVert X_i \rVert^2 \to 0~. \]

\item[(b)] As $n \to \infty$,
\begin{align*}
    \frac{1}{n} \sum_{1 \leq i \leq n} X_i & \to E_Q[\tilde X] \\
    \frac{1}{n} \sum_{1 \leq i \leq n} X_iX_i' & \to E_Q[\tilde X \tilde X']>0 \\
    \frac{1}{n} \sum_{1 \leq i \leq n} X_iY_i(d) & \to E_Q[\tilde X \tilde Y(d)] \text{ for } d \in \{0,1\} ~.
\end{align*}

\item[(c)] For any $d \in \{0,1\}$ and sequence $\{(i_j,i_j'): (i_j,i_j') \in \lambda_j^2, i_j \ne i_j', 1 \le j \leq m\}$, as $n \rightarrow \infty$,
\begin{align*}
\frac{1}{m} \sum_{1 \leq j \leq m} X_{i_j} X_{i_j'} &\stackrel{}{\rightarrow} E_Q[\tilde X\tilde X']>0 \\
\frac{1}{m} \sum_{1 \leq j \leq m} X_{i_j} Y_{i_j'}(d) &\stackrel{}{\rightarrow} E_Q[\tilde X\tilde Y(d)]~.
\end{align*}
\end{enumerate}
\end{assumption}

Assumption \ref{ass:super}(a) states that the finite population ``moments'' of the outcomes should converge to well-defined limit quantities; the condition holds almost surely by the strong law of large numbers if $W^{(n)}$ is modeled as an i.i.d.\ sample from the probability distribution $Q$. Assumption \ref{ass:super}(b) states that the average products of the outcomes in stratum-pairs should converge to well-defined limit quantities in such a way that their covariate values are also being matched in the limit; this condition can also be justified if $W^{(n)}$ is modeled as an i.i.d.\ sample from $Q$ and the stratification $\Lambda_n$ satisfies the property that 
\begin{align}\label{eq:close}
\frac{1}{m}\sum_{1 \le j \le \lfloor m/2 \rfloor}\max_{i, i' \in (\lambda_{2j - 1} \cup \lambda_{2j})}\|X_i - X_{i'}\|^2 \xrightarrow{P} 0~,
\end{align}
following similar arguments to those used in Lemma C.2 in \cite{bai2024inference}. In words, \eqref{eq:close} requires that the average of the squared distances of the covariate values in a paired stratum converges to zero asymptotically. This property is guaranteed by specific matching algorithms under appropriate assumptions; see, for instance, \cite{bai2022inference} and \cite{cytrynbaum2021optimal} for details. Assumption \ref{ass:X} introduces similar conditions to Assumption \ref{ass:super} but also requires that these conditions hold for the covariate vectors $X^{(n)}$. Note in particular that $E_Q[\tilde X \tilde Y(d)] = E_Q[\tilde X E_Q[\tilde Y(d) | \tilde X]]$ by the law of iterated expectations.

\begin{remark}\label{rem:equi}
Another setting in which Assumptions \ref{ass:super}--\ref{ass:X} are naturally satisfied is as follows. Let $X_i = \frac{i}{n}$ for $1 \leq i \leq n$, so that $X^{(n)}$ are equi-distant in the unit interval. Further assume $Y_i(d) = f_d(X_i)$ for $d \in \{0, 1\}$, where $f_0$ and $f_1$ are deterministic Lipschitz functions. Define $Q$ such that $\tilde{X} \sim \mathrm{Unif}[0, 1]$ under $Q$ and $\tilde Y(d) = f_d(\tilde X)$ degenerately for $d \in \{0, 1\}$. Then, it can be shown using elementary properties of the Riemann integral that Assumptions \ref{ass:super}--\ref{ass:X} hold.  More generally, an alternative limiting thought experiment that would yield similar results allows $W^{(n)}$ to be a realization of a sample taken without replacement from a finite population that is converging (in distribution) to $Q$.
\end{remark}

Using Assumption \ref{ass:super}, we have the following expression for the limit of $\var[\hat{\Delta}_n]$ after normalizing appropriately:

\begin{theorem}\label{thm:limit_V}
Under Assumption \ref{ass:super},
\[n\cdot\var[\hat{\Delta}_n] \rightarrow V~,\]
where
\[V := E_Q\left[\frac{\var_Q[\tilde{Y}(1)|\tilde{X}]}{\eta} + \frac{\var_Q[\tilde{Y}(0)|\tilde{X}]}{1- \eta}\right] - E_Q
\left[\var_Q[\tilde{Y}(1) - \tilde{Y}(0) | \tilde{X}]\right]~.\]
\end{theorem}
We can now motivate the construction of $\hat V_n$ in Section \ref{sec:main_var} through Theorem \ref{thm:limit_V}.  For the reasons outlined in Remark \ref{rem:better_bounds}, we focus on estimating
\begin{equation} \label{eq:Vobs}
V^{\rm obs} := E_Q\left[\frac{\var_Q[\tilde{Y}(1)|\tilde{X}]}{\eta} + \frac{\var_Q[\tilde{Y}(0)|\tilde{X}]}{1- \eta}\right]~,    
\end{equation}
which is an upper bound for $V$. The construction of $\hat{V}_n$ can be motivated following arguments similar to those used in our earlier work on the analysis of finely stratified experiments from a super-population perspective \citep[][]{bai2022inference,bai2024inference-b,bai2024inference,bai2025efficiency}. First, by expanding the square, it can be shown under Assumption \ref{ass:super} that
\begin{align*}
k \cdot \hat{\tau}_n^2 & \stackrel{P}{\to} \frac{1}{\eta} E_Q[\var_Q[\tilde Y(1) | \tilde X]] + \frac{1}{1 - \eta} E_Q[\var_Q[\tilde Y(0) | \tilde X]] + k E_Q[E_Q[\tilde Y(1) - \tilde Y(0) | \tilde X]^2] \\
& = V^{\rm obs} + k E_Q[E_Q[\tilde Y(1) - \tilde Y(0) | \tilde X]^2]~.
\end{align*}
It thus remains to construct a consistent estimator of the final term and subtract it from $\hat{\tau}^2_n$. For this purpose, it can also be shown under Assumption \ref{ass:super} that
\begin{equation} \label{eq:kappa}
k \cdot \hat \kappa_n \stackrel{P}{\to} k E_Q[E_Q[\tilde Y(1) - \tilde Y(0) | \tilde X]^2]~. 
\end{equation}
To understand why it is natural to expect $\hat \kappa_n$ to satisfy \eqref{eq:kappa} under Assumption \ref{ass:super}(b), note that $\hat \kappa_n$ is the average of products of the differences in means in a pair of strata.  Each pair of differences in means are independent conditional on the stratification.  Furthermore, for a sufficiently good stratification, each of them has conditional expectation close to $E_Q[\tilde Y(1) - \tilde Y(0) | \tilde X]$ for the same value of $\tilde X$.  For this reason, we expect $\hat \kappa_n$ to converge in the desired way.

In the remainder of the section we compare the asymptotic variance $V$ to the (appropriately scaled) limit of $\hat{V}_n$ and to the limits of the variance estimators proposed in \cite{imai2008variance} (i.e., $\hat{V}^{\rm IM}_n$), as well as the variance estimator proposed by \cite{fogarty2018mitigating}, which we now describe. Let $R$ be an $m \times L$ matrix for $L < m$ with rank $L$ and $H_R = R (R' R)^{-1} R'$ denote the projection matrix onto its column space. For each $R$, consider the estimator given by
\begin{equation} \label{proj_var_est}
\hat V_n^{\rm F}(R) =\frac{1}{m^2}\hat{\delta}_n' (\diag(I-H_R))^{-1/2}(I-H_R)(\diag(I-H_R))^{-1/2}\hat{\delta}_n,
\end{equation}
where $\hat{\delta}_n=(\hat \Delta_{1, n},\dots,\hat \Delta_{m, n})'$. Note that $\hat{V}_n^{\rm F}(\iota_m) = \hat{V}_n^{\rm IM}$, where $\iota_m$ is an $m \times 1$ vector of ones, so that $\hat{V}_n^{\rm F}(R)$ is indeed a generalization of $\hat{V}^{\rm IM}_n$. \cite{fogarty2018mitigating} focuses on the estimator given by $\hat V_n^{\rm F} = \hat V_n^{\rm F}(Q_2)$, where
\[ Q_2 = \begin{pmatrix}
     & (\bar X_{1, n} - \mu_{X, n})' \\
    \iota_m & \vdots \\
     & (\bar X_{m, n} - \mu_{X, n})'
\end{pmatrix}\] 
for $\bar X_{j, n} = \frac{1}{k} \sum_{i \in \lambda_j} X_i$ and $\mu_{X, n} = \frac{1}{n} \sum_{1 \leq i \leq n} X_i$. In particular, \cite{fogarty2018mitigating} argues that $\hat{V}_n^{\rm F}$ can be less conservative than $\hat{V}_n^{\rm IM}$ whenever the covariates are predictive of the treatment effects in an appropriate sense. The bias for $\hat V_n^{\rm F}$ is \citep[see][]{fogarty2018mitigating}
\[ E[\hat V_n^{\rm F}] - \var[\hat \Delta_n] = \frac{1}{m^2} \delta_n' (\diag(I-H_{Q_2}))^{-1/2}(I-H_{Q_2})(\diag(I-H_{Q_2}))^{-1/2} \delta_n \geq 0~,  \]
where $\delta_n = (\Delta_{1, n}, \dots, \Delta_{m, n})'$, so $\hat V_n^{\rm F}$ is also an upward-biased estimator for $\var[\hat \Delta_n]$.

 As discussed in Section \ref{sec:main}, it is possible to establish under weak assumptions on the sequence of population ``moments'' that each of these variances is consistent for their expectations in the sense of \eqref{eq:consistent}. Accordingly, in what follows it suffices to study the limit of the expectations of these estimators under Assumption \ref{ass:super}, and compare them to $V$.

\begin{theorem}\label{thm:main_var}
Under Assumption \ref{ass:super},
\begin{align*}
n \cdot E[\hat{V}_n] & \to V^{\rm obs}  \\
n \cdot E[\hat{V}_n^{\rm IM}] & \to V^{\rm obs} + k \cdot \var_Q[E_Q[\tilde{Y}(1) - \tilde{Y}(0)|\tilde{X}]] =: V^{\rm IM}~,
\end{align*}
for $V^{\rm obs}$ defined in \eqref{eq:Vobs}. Under Assumptions \ref{ass:super} and \ref{ass:X},
\[ n \cdot E[\hat{V}^{\rm F}_n] \to V^{\rm obs} + k \cdot \var_Q \Big [ E_Q[\tilde{Y}(1)-\tilde{Y}(0)|\tilde{X}] - \blp_Q(\tilde{Y}(1)-\tilde{Y}(0)|1, \tilde{X}) \Big ] =: V^{\rm F} \leq V^{\rm IM}~, \]
where $\blp$ denotes the best linear predictor, i.e., the unique solution under $Q$ such that $E[(\tilde Y(1) - \tilde Y(0) -  \blp_Q(\tilde{Y}(1)-\tilde{Y}(0)|1, \tilde{X})) (1, \tilde X)'] = 0$. Furthermore, if \eqref{eq:pop_moments}--\eqref{eq:consistent} hold, then for $\alpha \in (0, 1)$,
\begin{align*}
&P \left\{\hat \Delta_n - \sqrt{\hat V_n}\cdot z_{1 - \alpha / 2} \le  \Delta_n \le \hat \Delta_n + \sqrt{\hat V_n}\cdot z_{1 - \alpha / 2} \right\}   \to 1 - 2 \Phi(z_{\alpha / 2} \varsigma) \\
&P \left\{\hat \Delta_n - \sqrt{\hat V^{\rm IM}_n}\cdot z_{1 - \alpha / 2} \le  \Delta_n \le \hat \Delta_n + \sqrt{\hat V_n^{\rm IM}}\cdot z_{1 - \alpha / 2} \right\}  \to 1 - 2 \Phi(z_{\alpha / 2} \varsigma^{\rm IM}) \\
&P \left\{\hat \Delta_n - \sqrt{\hat V^{\rm F}_n}\cdot z_{1 - \alpha / 2} \le  \Delta_n \le \hat \Delta_n + \sqrt{\hat V_n^{\rm F}}\cdot z_{1 - \alpha / 2} \right\}  \to 1 - 2 \Phi(z_{\alpha / 2} \varsigma^{\rm F}) 
\end{align*}
where $\varsigma = (V / V^{\rm obs})^{1/2}$, $\varsigma^{\rm IM} = (V / V^{\rm IM})^{1/2}$, $\varsigma^{\rm F} = (V / V^{\rm F})^{1/2}$, and $\varsigma^{\rm IM} \leq \varsigma^{\rm F} \leq \varsigma \leq 1$.
\end{theorem}
Theorem \ref{thm:main_var} demonstrates that $\hat{V}_n$ converges to our desired upper bound $V^{\rm obs}$, but that in general $\hat{V}^{\rm IM}_n$ and $\hat{V}^{\rm F}_n$ do not. Specifically, $\hat{V}^{\rm IM}_n$ only attains $V^{\rm obs}$ when treatment effects are sufficiently homogeneous, in the sense that $E_Q[\tilde{Y}(1) - \tilde{Y}(0)|\tilde{X}]$ is constant; $\hat{V}^{\rm F}_n$ only attains $V^{\rm obs}$ when the conditional average treatment effect $E_Q[\tilde{Y}(1) - \tilde{Y}(0)|\tilde X]$ is linear and therefore equal to the best linear predictor of $\tilde{Y}(1)-\tilde{Y}(0)$ given $1$ and $\tilde{X}$. As a consequence, Theorem \ref{thm:main_var} demonstrates that, although confidence intervals constructed based on these variance estimators all cover $\Delta_n$ with probability at least $1 - \alpha$ in the limit, the confidence interval based on $\hat{V}_n$ is the least conservative whenever Assumptions \ref{ass:super} (and \ref{ass:X}) are satisfied. We further illustrate this phenomenon via simulation in Section \ref{sec:sims}. 

\begin{remark}
Our focus in this paper has been primarily on settings where $\min\{\ell, k - \ell\} = 1$, since, as explained in Section \ref{sec:review}, this case poses particular analytical challenges. We note however that $\hat{V}_n$ is still well-defined even in the case where $\min\{\ell, k - \ell\} > 1$, and Theorems \ref{thm:main} and \ref{thm:main_var} apply equally well to this setting. However, in this case, we note that the estimation strategy proposed in \eqref{eq:var_coarse} also provides valid inference, and it can be shown under appropriate assumptions that its corresponding probability limit is also $V^{\rm obs}$.
\end{remark}

% \begin{remark} \label{rem:cauchy}
% Note that
% \[ - E[\var[\tilde Y(1) - \tilde Y(0) | \tilde X]] = - E[\var[\tilde Y(1) | \tilde X]] - E[\var[
% \tilde Y(0) | \tilde X]] + 2 E[\mathrm{Cov}[\tilde Y(1), \tilde Y(0) | \tilde X]]~, \]
% and it follows from the Cauchy-Schwarz inequality that
% \[ E[\mathrm{Cov}[\tilde Y(1), \tilde Y(0) | \tilde X]] \leq E[\var[\tilde Y(1) | \tilde X]]^{1/2} E[\var[\tilde Y(0) | \tilde X]]^{1/2}~. \]
% Therefore,
% \begin{align*}
% V & \leq \frac{1}{k} \left ( E\left[\frac{\var_P[\tilde{Y}(1)|\tilde{X}]}{\eta} + \frac{\var_P[\tilde{Y}(0)|\tilde{X}]}{1- \eta}\right] - (E[\var[\tilde Y(1) | \tilde X]]^{1/2} - E[\var[
% \tilde Y(0) | \tilde X]]^{1/2})^2 \right )~,
% % & = \frac{1}{k} \left ( \left ( \frac{1 - \eta}{\eta} \right )^{1/2} E[\var[Y(1) | X]]^{1/2} + \left ( \frac{\eta}{1 - \eta} \right )^{1/2} E[\var[Y(0) | X]]^{1/2} \right )^2~.
% \end{align*}
% which is weakly smaller than the ``identified'' component of $V$. Based on this observation, we construct a variance estimator that is less conservative than \eqref{def::hat_V_n_mp} by substracting off $(E[\var[\tilde Y(1) | \tilde X]]^{1/2} - E[\var[\tilde Y(0) | \tilde X]]^{1/2})^2$ from any consistent estimator for the ``identified'' component of $V$.
% \end{remark}

\section{Simulations} \label{sec:sims}
In this section we examine the finite-$n$ behavior of several confidence intervals of the form described in \eqref{eq:valid}, constructed using different variance estimators. The variance estimators we consider are our proposed estimator $\hat{V}_n$, the \cite{imai2008variance} estimator $\hat{V}^{\rm IM}_n$, the \cite{fogarty2018mitigating} estimator $\hat{V}^{\rm F}_n$, and an alternative variance estimator $\hat{V}^{\rm alt}_n$ which can also be motivated using Theorem \ref{thm:limit_V} in our limiting thought experiment. Importantly, however, $\hat{V}^{\rm alt}_n$ is \emph{not} necessarily upward-biased under Assumption \ref{ass:dn} alone. As a consequence, we will demonstrate that the resulting confidence interval will not appropriately cover $\Delta_n$ whenever the limiting thought experiment described in Assumption \ref{ass:super} fails to hold. To construct $\hat{V}^{\rm alt}_n$, consider the following decomposition:
\begin{align*}
E_Q[\var_Q[\tilde{Y}(d)|\tilde{X}]] & = E_Q[\tilde{Y}(d)^2] - E_Q[E_Q[\tilde{Y}(d)|\tilde{X}]^2] \\
& = \var_Q[\tilde{Y}(d)] + E_Q[\tilde{Y}(d)]^2 - E_Q[E_Q[\tilde{Y}(d)|\tilde{X}]^2]~.    
\end{align*}
Under Assumption \ref{ass:super}(a), straightforward consistent estimators of $\var_Q[\tilde{Y}(d)]$ and $E_Q[\tilde{Y}(d)]^2$ are given by $\hat{\sigma}^2_n(d)$ and $\hat{\mu}_n(d)^2$, where
\begin{align*}
\hat{\mu}_n(d) & = \frac{1}{n(d)}\sum_{1 \le i \le n}Y_iI\{D_i = d\} \\
\hat{\sigma}^2_n(d) & = \frac{1}{n(d)}\sum_{1 \le i \le n}(Y_i - \hat{\mu}_n(d))^2I\{D_i = d\}~.
\end{align*}
Under Assumption \ref{ass:super}(b), an estimator of $E_Q[E_Q[\tilde{Y}(d)|\tilde{X}]^2]$ is 
\[ \hat \varsigma_n(d) =  \frac{2}{m} \sum\limits_{1 \leq j \leq \lfloor\frac{m}{2}\rfloor} {\sum_{i \in \lambda_{2j}, i' \in \lambda_{2j - 1}}} Y_iY_{i'}I\{D_i = D_{i'} = d\}~.\]
   An alternative estimator of $V^{\rm obs}$ under Assumption \ref{ass:super} is thus given by 

\[\hat{V}^{\rm alt}_n := \frac{1}{n}\left(\frac{\hat{\sigma}^2_n(1) + \hat{\mu}_n(1)^2 - \hat{\varsigma}_n(1)}{\eta} + \frac{\hat{\sigma}^2_n(0) + \hat{\mu}_n(0)^2  - \hat{\varsigma}_n(0)}{1 - \eta} \right)~.\]
We remark that, in the case of matched pairs, $\hat{V}_n^{\rm alt}$ also coincides with the estimator that would be obtained if we applied the estimation strategy in \eqref{eq:var_coarse} to the paired strata directly; see Remark \ref{rem:collapse} for further discussion.

We generate our population $((Y_i(1), Y_i(0), X_i): 1 \le i \le n)$ using i.i.d.\ draws from a probability distribution. The potential outcomes are generated according to the equation:
\begin{equation*}
Y_i(d) = \mu_d + \mu_d(X_i) + \epsilon_{di}~,
\end{equation*}
where $((X_i, \epsilon_{0i}, \epsilon_{1i}): 1\leq i \leq n)$ are i.i.d., $(X_i, \epsilon_{0i}, \epsilon_{1i})$ are independent with $X_i \sim U[0,1]$, $\epsilon_{di} \sim N(0, 1)$, the parameters $\mu_d$ are given by $\mu_1 = 0.25$, $\mu_0 = 0$, and $\mu_d(\cdot)$ are specified as
\begin{enumerate}[{\bf Model} 1:]
	\item $\mu_{0}(X_i) = 20(X_i - 1/2)$, $\mu_{1}(X_i) = 10(X_i - 1/2)$ 
	\item $\mu_0(X_i) = 40(X_i^2 - 4/3)$, $\mu_1(X_i) = 10(X_i^2 - 4/3)$~.
\end{enumerate}
For a population of size $n = 1000$, the potential outcomes and covariates are generated as i.i.d.\ draws \emph{once} and then fixed in repeated samples. When considering populations of size $n \in \{100,250,500,750\}$, we draw a subsample from our population of size $1000$ \emph{once} and then fix this in repeated samples. 

In each Monte Carlo iteration, we assign units to treatment using a matched pairs design (i.e., a finely stratified design with $\ell = 1$ and $k = 2$), under two alternative matching methods, which we call ``good match'' and ``bad match.'' For ``good match,'' we sort units in increasing values of $X_i$ and then pair adjacent units. For ``bad match,'' we sort units in increasing values of $X_i$ and then pair them such that the unit with the \emph{smallest} value of $X$ is matched with the unit with the \emph{largest} value of $X$, the second smallest with the second largest, and so on. In both cases, when re-arranging the strata to construct our paired strata $\{(\lambda_{2j-1}, \lambda_{2j}): 1 \le j \le \lfloor m/2 \rfloor\}$, we sort strata in increasing values of their covariate \emph{averages} and then pair adjacent pairs. As a result, the ``good match'' design minimizes the average squared distances of the covariate values in paired strata and can be shown to satisfy the conditions of our limiting thought experiment in Assumptions \ref{ass:super} and \ref{ass:X}. In contrast, the ``bad match'' design is constructed specifically to fail the conditions of the limiting thought experiment.

Tables \ref{tab:sims_linear} and \ref{tab:sims_quadratic} report the coverage and average length of $95\%$ confidence intervals computed across 5000 Monte Carlo replications. When matches are good, we see that all estimators have appropriate (albeit conservative) coverage. In Table \ref{tab:sims_linear} with good matches, the average length produced by $\hat{V}^{\rm IM}$ is largest while the average lengths produced by the other three estimators are comparable; this is in line with the theoretical results presented in Theorem \ref{thm:main_var} given that the conditional average treatment effect is linear under Model 1. In Table \ref{tab:sims_quadratic} with good matches, the average lengths produced by $\hat{V}_n$ and $\hat{V}_n^{\rm alt}$ are shortest, again in line with the theoretical results presented in Theorem \ref{thm:main_var}.

When matches are bad, the conditions in Theorem \ref{thm:main_var} no longer apply, and instead we must rely on the asymptotic validity provided by Theorem \ref{thm:upward} alone. With this in mind we see that in both models, the estimators $\hat{V}_n^{\rm IM}$, $\hat{V}_n^{\rm F}$ and $\hat{V}_n$ all provide appropriate coverage; this is in line with our theoretical results given that all three estimators are upward-biased under Assumption \ref{ass:dn}. The average lengths of the confidence intervals based on these estimators are comparable. On the other hand, $\hat{V}_n^{\rm alt}$, which is not guaranteed to be upward-biased in general, undercovers even in large populations. 

\begin{table}[ht]
    \centering
    \begin{tabular}{cccccc}
\toprule
        \textbf{Population size} & \textbf{Estimator} &
        \multicolumn{2}{c}{\textbf{Good matches}} &
        \multicolumn{2}{c}{\textbf{Bad matches}} \\
        \cmidrule(lr){3-4} \cmidrule(lr){5-6}
        & & Coverage & Length & Coverage & Length \\
        \hline
\multirow{4}{*}{100}
& $\hat{V}_n^{\rm IM}$ & 1.000 & 1.963 & 0.944 & 4.681 \\
& $\hat{V}_n^{\rm F}$ & 0.985 & 0.842 & 0.938 & 4.671 \\
& $\hat{V}_n$ & 0.976 & 0.793 & 0.929 & 4.634 \\
& $\hat{V}_n^{\rm alt}$ & 0.989 & 0.859 & 0.844 & 3.547 \\
\hline
\multirow{4}{*}{250}
& $\hat{V}_n^{\rm IM}$ & 1.000 & 1.189 & 0.946 & 3.018 \\
& $\hat{V}_n^{\rm F}$ & 0.989 & 0.511 & 0.944 & 3.016 \\
& $\hat{V}_n$ & 0.992 & 0.534 & 0.942 & 3.005 \\
& $\hat{V}_n^{\rm alt}$ & 0.998 & 0.618 & 0.863 & 2.281 \\
\hline
\multirow{4}{*}{500}
& $\hat{V}_n^{\rm IM}$ & 1.000 & 0.825 & 0.952 & 2.155 \\
& $\hat{V}_n^{\rm F}$ & 0.993 & 0.368 & 0.952 & 2.155 \\
& $\hat{V}_n$ & 0.992 & 0.369 & 0.951 & 2.153 \\
& $\hat{V}_n^{\rm alt}$ & 0.992 & 0.361 & 0.862 & 1.628 \\
\hline
\multirow{4}{*}{750}
& $\hat{V}_n^{\rm IM}$ & 1.000 & 0.651 & 0.950 & 1.745 \\
& $\hat{V}_n^{\rm F}$ & 0.994 & 0.293 & 0.950 & 1.745 \\
& $\hat{V}_n$ & 0.994 & 0.294 & 0.948 & 1.743 \\
& $\hat{V}_n^{\rm alt}$ & 0.997 & 0.312 & 0.858 & 1.316 \\
\hline
\multirow{4}{*}{1000}
& $\hat{V}_n^{\rm IM}$ & 1.000 & 0.565 & 0.949 & 1.518 \\
& $\hat{V}_n^{\rm F}$ & 0.995 & 0.247 & 0.949 & 1.518 \\
& $\hat{V}_n$ & 0.995 & 0.245 & 0.950 & 1.516 \\
& $\hat{V}_n^{\rm alt}$ & 0.996 & 0.246 & 0.857 & 1.143 \\
\bottomrule
    \end{tabular}
        \caption{Coverage probabilities and average length (Model 1)}
\label{tab:sims_linear}
\end{table}

\begin{table}[ht]
    \centering
    \begin{tabular}{cccccc}
\toprule
        \textbf{Population size} & \textbf{Estimator} &
        \multicolumn{2}{c}{\textbf{Good matches}} &
        \multicolumn{2}{c}{\textbf{Bad matches}} \\
        \cmidrule(lr){3-4} \cmidrule(lr){5-6}
        & & Coverage & Length & Coverage & Length \\
        \hline
\multirow{4}{*}{100}
& $\hat{V}_n^{\rm IM}$ & 1.000 & 5.288 & 0.944 & 8.185 \\
& $\hat{V}_n^{\rm F}$ & 1.000 & 1.389 & 0.939 & 8.173 \\
& $\hat{V}_n$ & 0.980 & 0.848 & 0.930 & 8.102 \\
& $\hat{V}_n^{\rm alt}$ & 0.990 & 0.910 & 0.880 & 6.813 \\
\hline
\multirow{4}{*}{250}
& $\hat{V}_n^{\rm IM}$ & 1.000 & 3.274 & 0.947 & 5.174 \\
& $\hat{V}_n^{\rm F}$ & 1.000 & 0.909 & 0.944 & 5.153 \\
& $\hat{V}_n$ & 1.000 & 0.812 & 0.940 & 5.133 \\
& $\hat{V}_n^{\rm alt}$ & 1.000 & 0.992 & 0.892 & 4.292 \\
\hline
\multirow{4}{*}{500}
& $\hat{V}_n^{\rm IM}$ & 1.000 & 2.256 & 0.952 & 3.587 \\
& $\hat{V}_n^{\rm F}$ & 1.000 & 0.675 & 0.952 & 3.580 \\
& $\hat{V}_n$ & 0.992 & 0.370 & 0.951 & 3.567 \\
& $\hat{V}_n^{\rm alt}$ & 0.992 & 0.363 & 0.905 & 3.003 \\
\hline
\multirow{4}{*}{750}
& $\hat{V}_n^{\rm IM}$ & 1.000 & 1.816 & 0.951 & 2.919 \\
& $\hat{V}_n^{\rm F}$ & 1.000 & 0.548 & 0.951 & 2.918 \\
& $\hat{V}_n$ & 1.000 & 0.359 & 0.950 & 2.913 \\
& $\hat{V}_n^{\rm alt}$ & 1.000 & 0.406 & 0.908 & 2.459 \\
\hline
\multirow{4}{*}{1000}
& $\hat{V}_n^{\rm IM}$ & 1.000 & 1.588 & 0.952 & 2.552 \\
& $\hat{V}_n^{\rm F}$ & 1.000 & 0.462 & 0.952 & 2.545 \\
& $\hat{V}_n$ & 0.995 & 0.245 & 0.950 & 2.537 \\
& $\hat{V}_n^{\rm alt}$ & 0.996 & 0.246 & 0.898 & 2.130 \\
\bottomrule
    \end{tabular}
    \caption{Coverage probabilities and average length (Model 2)}
    \label{tab:sims_quadratic}
\end{table}

\section{Recommendations for Empirical Practice}
Based on our theoretical results as well as the simulation study above, we conclude with some recommendations for empirical practice. Overall, we argue that our novel estimator $\hat{V}_n$ is well suited for performing design-based inference on the average treatment effect, particularly in settings with $\min\{\ell, k - \ell\} = 1$. Since the magnitude of the bias of this estimator depends explicitly on the quality of the pairing of the strata, we recommend that practitioners form the strata pairs by applying an optimal non-bipartite matching algorithm \citep[see, for instance,][]{greevy2004optimal} to the stratum-level covariate means $\bar{X}_j = \frac{1}{k}\sum_{i \in \lambda_j}X_i$; this algorithm is available in the {\tt R} package {\tt nbpMatching}.

An interesting feature of our results relative to our findings in prior work on the analysis of randomized experiments from a super-population perspective \citep[][]{bai2022inference,bai2024inference-b,bai2024inference,bai2025efficiency} is that, by appealing directly to Theorem \ref{thm:main}, asymptotically valid inference on $\Delta_n$ is possible under minimal assumptions on the nature of the matching procedure. In contrast, inference on the average treatment effect from a super-population perspective seems to require an assumption like \eqref{eq:close}. However, we note that inferences based on the design-based estimator $\hat{V}_n$ are \emph{not} guaranteed to be valid if we view the sample as being drawn from a larger (finite or super-) population. In such cases, it can be shown that the super-population variance estimators proposed in \cite{bai2022inference} \citep[see also][]{bai2024inference, bai2025efficiency} remain valid in the limiting thought experiment of Section \ref{sec:variances}.

\clearpage
\appendix

\section{Proofs of Main Results}

\subsection{Proof of Theorem \ref{thm:upward}}
First, we show that Assumptions \eqref{eq:pop_moments} and \eqref{eq:pop_nondegen} imply the conditions in Theorem \ref{thm:normal}, so that 
\begin{equation}\label{eq:normal}
\frac{\hat \Delta_n - \Delta_n}{\var[\hat \Delta_n]^{1/2}} \stackrel{d}{\rightarrow} N(0,1)~.
\end{equation}
To that end, we first show 
\begin{equation}\label{eq:num_zero}
\frac{1}{n}\max_{1 \leq j \leq m} \max_{i \in \lambda_j} (R_i - \bar R_j)^2 \rightarrow 0~,
\end{equation}
where $R_i$, $\bar{R}_j$ are defined in the statement of Theorem \ref{thm:normal}, and then show that 
$\frac{1}{n}\sum_{1 \leq j \leq m} \sum_{i \in \lambda_j} (R_i - \bar R_j)^2$ is bounded away from zero.
Note that from the inequality $(a + b)^2 \le 2(a^2 + b^2)$, 
\[\max_{i \in \lambda_j}(R_i - \bar{R}_j)^2 \lesssim \max_{i \in \lambda_j}\frac{1}{\eta^2}(Y_i(1) - \bar Y_{j, n}(1))^2 + \frac{1}{(1 - \eta)^2}\max_{i \in \lambda_j}(Y_i(0) - \bar Y_{j, n}(0))^2~.\]
Since by the triangle inequality $\max_{i \in \lambda_j}|Y_i(d) - \bar Y_{j, n}(d)| \le 2\max_{i \in \lambda_j}|Y_i(d)|$, 
\begin{align*}
\frac{1}{n}\max_{1 \leq j \leq m}\max_{i \in \lambda_j}(Y_i(d) - \bar Y_{j, n}(d))^2 &\lesssim \frac{1}{n}\max_{1 \leq i \leq n}Y_i(d)^2 \rightarrow 0~,
\end{align*}
where the final convergence follows from \eqref{eq:pop_moments}. \eqref{eq:num_zero} therefore follows. Next, since it is shown in the proof of Theorem \ref{thm:normal} that $\frac{1}{n}\sum_{1 \leq j \leq m} \sum_{i \in \lambda_j} (R_i - \bar R_j)^2$ is proportional to $n\cdot \var[\hat{\Delta}_n]$ we immediately obtain from Assumption \eqref{eq:pop_nondegen} that $\frac{1}{n}\sum_{1 \leq j \leq m} \sum_{i \in \lambda_j} (R_i - \bar R_j)^2$ is bounded away from zero. We thus obtain \eqref{eq:highlevel} and so \eqref{eq:normal} follows.
To complete the proof, we apply Lemma \ref{lem:ratio} with $X_n = \sqrt n(\hat \Delta_n - \Delta_n)$, $\sigma_n^2 = n \var[\hat \Delta_n]$, and $\hat \sigma_n^2 = n \hat V_n$. \qed

\subsection{Proof of Theorem \ref{thm:main}}
First we establish (a). By direct calculation using Assumption \ref{ass:dn} and Theorem 2.1 in \cite{cochran1977sampling}, we have
\begin{align*}
    & E[\hat \tau_n^2 - \hat \kappa_n] \\
    & = \frac{1}{n \ell} \sum_{1 \leq i \leq n} Y_i(1)^2 + \frac{2 (\ell - 1) }{n (k - 1) \ell} \sum_{1 \leq j \leq m} \sum_{i < i' \in \lambda_j} Y_i(1) Y_{i'}(1) \\
    & \hspace{3em} + \frac{1}{n (k - \ell)} \sum_{1 \leq i \leq n} Y_i(0)^2 + \frac{2 (k - \ell - 1) }{n (k - 1) (k - \ell)} \sum_{1 \leq j \leq m} \sum_{i < i' \in \lambda_j} Y_i(0) Y_{i'}(0) \\
    & \hspace{3em} - \frac{2}{n (k - 1)} \sum_{1 \leq j \leq m} \sum_{i \neq i' \in \lambda_j} Y_i(1) Y_{i'}(0) \\
    & \hspace{3em} - \frac{2}{m} \sum_{1 \leq j \leq \frac{m}{2}} (\bar Y_{2j - 1, n}(1) - \bar Y_{2j - 1, n}(0)) (\bar Y_{2j, n}(1) - \bar Y_{2j, n}(0))~.
\end{align*}
On the other hand,
\begin{align*}
    m\cdot\var[\hat{\Delta}_n] & = \frac{1}{n} \sum_{1 \leq j \leq m} \bigg ( \frac{k}{(k - 1) \ell} \sum_{i \in \lambda_j} (Y_i(1) - \bar Y_{j, n}(1))^2 + \frac{k}{(k - 1) (k - \ell)} \sum_{i \in \lambda_j} (Y_i(0) - \bar Y_{j, n}(0))^2 \\
    & \hspace{6em} - \frac{1}{k - 1} \sum_{i \in \lambda_j} (Y_i(1) - Y_i(0) - (\bar Y_{j, n}(1) - \bar Y_{j, n}(0)) )^2 \bigg ) \\
    & = \frac{1}{n} \sum_{1 \leq j \leq m} \bigg ( \frac{1}{\ell} \sum_{i \in \lambda_j} Y_i(1)^2 - \frac{2}{(k - 1) \ell} \sum_{i < i' \in \lambda_j} Y_i(1) Y_{i'}(1) \\
    & \hspace{6em} + \frac{1}{k - \ell} \sum_{i \in \lambda_j} Y_i(0)^2 - \frac{2}{(k - 1) (k - \ell)} \sum_{i < i' \in \lambda_j} Y_i(0) Y_{i'}(0) \\
    & \hspace{6em} - \frac{1}{k - 1} \Big ( \sum_{i \in \lambda_j} (Y_i(1) - Y_i(0))^2 - k (\bar Y_{j, n}(1) - \bar Y_{j, n}(0))^2 \Big ) \bigg )~.
\end{align*}
We therefore have
\begin{align*}
    & E[\hat \tau_n^2 - \hat \kappa_n] - m\cdot\var[\hat{\Delta}_n]\\
    & = \frac{2}{n (k - 1)} \sum_{1 \leq j \leq m} \sum_{i < i' \in \lambda_j} Y_i(1) Y_{i'}(1) + \frac{2}{n (k - 1)} \sum_{1 \leq j \leq m} \sum_{i < i' \in \lambda_j} Y_i(0) Y_{i'}(0) \\
    & \hspace{3em} - \sum_{1 \leq j \leq m} \frac{2}{n (k - 1)} \sum_{i \neq i' \in \lambda_j} Y_i(1) Y_{i'}(0) \\
    & \hspace{3em} + \frac{1}{n (k - 1)} \sum_{1 \leq j \leq m} \Big ( \sum_{i \in \lambda_j} (Y_i(1) - Y_i(0))^2 - k (\bar Y_{j, n}(1) - \bar Y_{j, n}(0))^2 \Big ) \\
    & \hspace{3em} - \frac{2}{m} \sum_{1 \leq j \leq \frac{m}{2}} (\bar Y_{2j - 1, n}(1) - \bar Y_{2j - 1, n}(0)) (\bar Y_{2j, n}(1) - \bar Y_{2j, n}(0)) \\
    & = \frac{1}{n (k - 1)} \sum_{1 \leq j \leq m} \Big ( k^2 \bar Y_{j, n}(1)^2 - \sum_{i \in \lambda_j} Y_i(1)^2 \Big ) + \frac{1}{n (k - 1)} \sum_{1 \leq j \leq m} \Big ( k^2 \bar Y_{j, n}(0)^2 - \sum_{i \in \lambda_j} Y_i(0)^2 \Big ) \\
    & \hspace{3em} - \frac{2}{n (k - 1)} \sum_{1 \leq j \leq m} \Big ( k^2 \bar Y_{j, n}(1) \bar Y_{j, n}(0) - \sum_{i \in \lambda_j} Y_i(1) Y_i(0) \Big ) \\
    & \hspace{3em} + \frac{1}{n (k - 1)} \sum_{1 \leq j \leq m} \Big ( \sum_{i \in \lambda_j} (Y_i(1) - Y_i(0))^2 - k (\bar Y_{j, n}(1) - \bar Y_{j, n}(0))^2 \Big ) \\
    & \hspace{3em} - \frac{2}{m} \sum_{1 \leq j \leq \frac{m}{2}} (\bar Y_{2j - 1, n}(1) - \bar Y_{2j - 1, n}(0)) (\bar Y_{2j, n}(1) - \bar Y_{2j, n}(0)) \\
    & = \frac{1}{n (k - 1)} \sum_{1 \leq i \leq n} (- Y_i(1)^2 - Y_i(0)^2 + 2 Y_i(1) Y_i(0) + Y_i(1)^2 + Y_i(0)^2 - 2 Y_i(1) Y_i(0)) \\
    & \hspace{3em} + \frac{1}{n (k - 1)} \sum_{1 \leq j \leq m} (k^2 \bar Y_{j, n}(1)^2 + k^2 \bar Y_{j, n}(0)^2 - 2 k^2 \bar Y_{j, n}(1) \bar Y_{j, n}(0) \\
    & \hspace{12em} - k \bar Y_{j, n}(1)^2 - k \bar Y_{j, n}(0)^2 + 2 k \bar Y_{j, n}(1) \bar Y_{j, n}(0) ) \\
    & \hspace{3em} - \frac{2}{m} \sum_{1 \leq j \leq \frac{m}{2}} (\bar Y_{2j - 1, n}(1) - \bar Y_{2j - 1, n}(0)) (\bar Y_{2j, n}(1) - \bar Y_{2j, n}(0)) \\
    & = \frac{1}{m} \sum_{1 \leq j \leq m} (\bar Y_{j, n}(1) - \bar Y_{j, n}(0))^2 - \frac{2}{m} \sum_{1 \leq j \leq \frac{m}{2}} (\bar Y_{2j - 1, n}(1) - \bar Y_{2j - 1, n}(0)) (\bar Y_{2j, n}(1) - \bar Y_{2j, n}(0)) \\
    & = \frac{1}{m} \sum_{1 \leq j \leq \frac{m}{2}} (\bar Y_{2j - 1, n}(1) - \bar Y_{2j - 1, n}(0) - (\bar Y_{2j, n}(1) - \bar Y_{2j, n}(0)))^2~.
\end{align*}
Next, we establish (b). We'll show that
\begin{align*}
|\hat{\tau}^2_n - E[\hat{\tau}^2_n]| & \xrightarrow{P} 0 \\
|\hat{\kappa}_n - E[\hat{\kappa}_n]| & \xrightarrow{P} 0~,
\end{align*}
from which the result will follow from the triangle inequality. 
By Chebyshev's inequality, it suffices to show that $\var[\hat{\tau}^2_n] \rightarrow 0$ and $\var[\hat{\kappa}_n] \rightarrow 0$. To that end, define
\[\hat{\Delta}_j = \frac{1}{\ell}\sum_{i \in \lambda_j}Y_iD_i - \frac{1}{k - \ell}\sum_{i \in \lambda_j}Y_i(1 - D_i)~.\]
Then,
\begin{align}
\label{eq:tau_var} \var[\hat{\tau}_n^2] & = \frac{1}{m^2}\sum_{1 \leq j \leq m}\var[\hat{\Delta}_{j, n}^2] \\
\label{eq:kappa_var} \var[\hat{\kappa}_n] & = \frac{4}{m^2}\sum_{1 \le j \le \frac{m}{2}}\var[\hat{\Delta}_{2j-1, n}\hat{\Delta}_{2j, n}]~.
\end{align}
First, we bound \eqref{eq:tau_var}:
\begin{align*}
\frac{1}{m^2}\sum_{1 \leq j \leq m}\var[\hat{\Delta}_j^2] &\le \frac{1}{m^2}\sum_{1 \leq j \leq m}E[\hat{\Delta}_j^4] \\
&\lesssim \frac{1}{m^2}\sum_{1 \leq j \leq m}\left(E\left[\left(\frac{1}{\ell}\sum_{i \in \lambda_j}Y_iD_i\right)^4\right] + E\left[\left(\frac{1}{k - \ell}\sum_{i \in \lambda_j}Y_i(1 - D_i)\right)^4\right]\right) \\
&\le \frac{1}{m^2}\sum_{1 \leq j \leq m}\left(E\left[\frac{1}{\ell}\sum_{i \in \lambda_j}Y_i^4D_i\right] + E\left[\frac{1}{k - \ell}\sum_{i \in \lambda_j}Y_i^4(1 - D_i)\right]\right) \\
& = \frac{1}{m^2}\sum_{1 \leq j \leq m}\left(\frac{1}{k}\sum_{i \in \lambda_j}Y_i(1)^4 + \frac{1}{k}\sum_{i \in \lambda_j}Y_i(0)^4\right) \\ 
& = \frac{1}{m}\left(\frac{1}{n}\sum_{1 \leq i \leq n}Y_i(1)^4 + \frac{1}{n}\sum_{1 \leq i \leq n}Y_i(0)^4\right) \rightarrow 0~,
\end{align*}
where the first inequality follows from the definition of the variance, the second from repeated applications of the inequality $(a \pm b)^2 \le 2(a^2 + b^2)$, the third from Jensen's inequality, the first equality from Assumption \ref{ass:dn} and Theorem 2.1 in \cite{cochran1977sampling}, which together imply $E[D_i] = \ell / k$, and the final line by \eqref{eq:pop_moments}.

Next, we bound \eqref{eq:kappa_var}:
\begin{align*}
 \frac{4}{m^2}\sum_{1 \le j \le \frac{m}{2}}\var[\hat{\Delta}_{2j-1,n}\hat{\Delta}_{2j,n}] &\le \frac{2}{m^2}\sum_{1 \le j \le \frac{m}{2}}E[2\hat{\Delta}_{2j-1,n}^2\hat{\Delta}_{2j,n}^2] \\
 &\le \frac{2}{m^2}\sum_{1 \le j \le \frac{m}{2}}E[\hat{\Delta}_{2j-1,n}^4 + \hat{\Delta}_{2j,n}^4] \\ 
 & \lesssim \frac{2}{m^2}\sum_{1 \le j \le \frac{m}{2}}\left(\frac{1}{k}\sum_{i \in \lambda_{2j - 1}}\sum_{d \in \{0,1\}}Y_i(d)^4 + \frac{1}{k}\sum_{i \in \lambda_{2j}}\sum_{d \in \{0, 1\}}Y_i(d)^4\right) \\ 
 & = \frac{2}{m}\left(\frac{1}{n}\sum_{1 \leq i \leq n}Y_i(1)^4 + \frac{1}{n}\sum_{1 \leq i \leq n}Y_i(0)^4\right) \rightarrow 0~.
\end{align*}
where the first inequality follows from the definition of the variance, the second from the inequality $2|ab| \le (a^2 + b^2)$, the third from following the same argument as when bounding \eqref{eq:tau_var}, and the final line by Assumption \eqref{eq:pop_moments}. \qed

\subsection{Proof of Theorem \ref{thm:limit_V}}

From the proof of Theorem \ref{thm:normal}, we obtain that
\[ n \cdot \var[\hat{\Delta}_n] = k \sum_{1 \le j \le m}\var[A_j] = k \eta^2\left(\frac{1}{\ell} - \frac{1}{k}\right) \frac{1}{k-1} \frac{1}{m} \sum_{1 \le j \le m} \sum_{i \in \lambda_j}(R_i - \bar{R}_j)^2 \]
Consider the following expansion of $\frac{1}{m}\sum_{1 \le j \le m} \sum_{i \in \lambda_j}\bar{R}_j^2$:
\begin{align*}
    \frac{1}{m}\sum_{1 \le j \le n} \sum_{i \in \lambda_j}\bar{R}_j^2 = \frac{k}{m} \sum_{1 \le j \le m} \bar{R}_j^2 
    & = \frac{k}{m k^2} \sum_{1 \le j \le m} \sum_{i, i' \in \lambda_j} R_i R_{i'} \\
    & = \frac{1}{km} \sum_{1 \leq i \leq km} R_i^2 + \frac{2}{km} \sum_{1 \le j \le m} \sum_{i < i' \in \lambda_j} R_i R_{i'} \\
    & = \frac{1}{n} \sum_{1 \le i \le n} R_i^2 + \frac{k-1}{m} \sum_{1 \le j \le m} \frac{1}{\binom{k}{2}} \sum_{i < i' \in \lambda_j} R_i R_{i'}~.
\end{align*}
We now show that
\begin{align}
\label{eq:R1} \frac{1}{n} \sum_{1 \leq i \leq n} R_i^2 & \to E_Q[\tilde R^2] \\
\label{eq:R2} \frac{1}{m} \sum_{1 \le j \le m} \frac{1}{\binom{k}{2}} \sum_{i < i' \in \lambda_j} R_i R_{i'} & \to E_Q[E_Q[\tilde R | \tilde X]^2]~,
\end{align}
where $\tilde{R}:= \frac{\tilde{Y}(1)}{\eta} + \frac{\tilde{Y}(0)}{1 - \eta}$. Note \eqref{eq:R1} follows from Assumption \ref{ass:super}(a) immediately. To see how \eqref{eq:R2} is implied by Assumption \ref{ass:super}(b), for each $j$, let $i_{j, 1} < \dots < i_{j, k}$ be such that $\lambda_j = \{i_{j, 1}, \dots, i_{j, k}\}$, and note Assumption \ref{ass:super}(b) implies that for each $1 \leq \gamma_1 < \gamma_2 \leq k$,
\[ \frac{1}{m} \sum_{1 \leq j \leq m} R_{i_{j, \gamma_1}} R_{i_{j, \gamma_2}} \to E_Q[E_Q[\tilde R | \tilde X]^2]~. \]
The result in \eqref{eq:R2} then follows upon noting
\[ \frac{1}{m} \sum_{1 \leq j \leq m} \frac{1}{\binom{k}{2}} \sum_{i < i' \in \lambda_j} R_i R_{i'} = \frac{1}{\binom{k}{2}} \sum_{1 \leq \gamma_1 < \gamma_2 \le k} \frac{1}{m} \sum_{1 \leq j \leq m} R_{i_{j, \gamma_1}} R_{i_{j, \gamma_2}}~. \]
We thus obtain that
\[\frac{1}{m}\sum_{1 \le j \le m}\sum_{i \in \lambda_j}(R_i - \bar{R}_j)^2 = \frac{1}{m}\sum_{1 \le j \le m} \Big ( \sum_{i \in \lambda_j} R_i^2 - k \bar{R}_j^2 \Big ) \to (k-1) E_Q[\var_Q[\tilde{R} | \tilde{X}]]~.\]
Finally,
\begin{align*}
     n \cdot \var[\hat{\Delta}_n] = k \sum_{1 \le j \le m}\var[A_j] & \to k\eta^2\left(\frac{1}{\ell}- \frac{1}{k}\right)E_Q[\var_Q[\tilde{R}|\tilde{X}]]  \\ 
     & =\eta(1-\eta)E_Q[\var_Q[\tilde{R}|\tilde{X}]] = \eta(1-\eta)E_Q\left[\var_Q\left[\frac{\tilde{Y}(1)}{\eta} + \frac{\tilde{Y}(0)}{1- \eta}\Bigg|\tilde{X}\right]\right] \\
     & =  E_Q\left[\frac{\var_Q[\tilde{Y}(1)|\tilde{X}]}{\eta} + \frac{\var_Q[\tilde{Y}(0)|\tilde{X}]}{1- \eta}\right] - E_Q[\var_Q[\tilde{Y}(1) - \tilde{Y}(0) | \tilde{X}]]~,
\end{align*}
as desired. \qed

\subsection{Proof of Theorem \ref{thm:main_var}}

(a) We start by deriving the limit of $n E[\hat V_n]$. Recall from the proof of Theorem \ref{thm:main} that
\begin{align*}
     E[\hat \tau_n^2]
    & = \frac{1}{n \ell} \sum_{1 \leq i \leq n} Y_i(1)^2 + \frac{2 (\ell - 1) }{n (k - 1) \ell} \sum_{1 \leq j \leq m} \sum_{i < i' \in \lambda_j} Y_i(1) Y_{i'}(1) \\
    & \hspace{3em} + \frac{1}{n (k - \ell)} \sum_{1 \leq i \leq n} Y_i(0)^2 + \frac{2 (k - \ell - 1) }{n (k - 1) (k - \ell)} \sum_{1 \leq j \leq m} \sum_{i < i' \in \lambda_j} Y_i(0) Y_{i'}(0) \\
    & \hspace{3em} - \frac{2}{n (k - 1)} \sum_{1 \leq j \leq m} \sum_{i \neq i' \in \lambda_j} Y_i(1) Y_{i'}(0) \\
   & = \frac{1}{m} \sum_{1 \leq j \leq m} \bigg ( \frac{1}{\ell} \frac{1}{k} \sum_{i \in \lambda_j} Y_i^2(1) + \frac{\ell - 1}{\ell} \frac{1}{\binom{k}{2}} \sum_{i < i' \in \lambda_j} Y_i(1) Y_{i'}(1) \bigg ) \\
& \hspace{2em} + \frac{1}{m} \sum_{1 \leq j \leq m} \bigg ( \frac{1}{k - \ell} \frac{1}{k} \sum_{i \in \lambda_j} Y_i^2(0) + \frac{k - \ell - 1}{k - \ell} \frac{1}{\binom{k}{2}} \sum_{i < i' \in \lambda_j} Y_i(0) Y_{i'}(0) \bigg ) \\
& \hspace{2em} - \frac{1}{m} \sum_{1 \leq j \leq m} \frac{1}{\binom{k}{2}} \sum_{i \neq i' \in \lambda_j} Y_i(1) Y_{i'}(0)~.
\end{align*}
Assumption \ref{ass:super} implies that
\begin{align}
\label{eq:tau1} \frac{1}{n} \sum_{1 \leq i \leq n} Y_i(d)^2 & \to E_Q[\tilde Y^2(d)] \text{ for } d \in \{0, 1\} \\
\label{eq:tau2} \frac{1}{m} \sum_{1 \leq j \leq m} \frac{1}{\binom{k}{2}} \sum_{i < i' \in \lambda_j} Y_i(d) Y_{i'}(d) & \to E_Q[\tilde E_Q[\tilde Y(d) | \tilde X]^2] \text{ for } d \in \{0, 1\} \\
\label{eq:tau3} \frac{1}{m} \sum_{1 \leq j \leq m} \frac{1}{\binom{k}{2}} \sum_{i \neq i' \in \lambda_j} Y_i(1) Y_{i'}(0) & \to 2 E_Q[E_Q[\tilde Y(1) | \tilde X] E_Q[\tilde Y(0) | \tilde X]]~,
\end{align}
where \eqref{eq:tau1} follows from Assumption \ref{ass:super}(a) and \eqref{eq:tau2}--\eqref{eq:tau3} follow from Assumption \ref{ass:super}(b). To see how \eqref{eq:tau2} is implied by Assumption \ref{ass:super}(b), for each $j$, let $i_{j, 1} < \dots < i_{j, k}$ be such that $\lambda_j = \{i_{j, 1}, \dots, i_{j, k}\}$, and note Assumption \ref{ass:super}(b) implies that for each $1 \leq \gamma_1 < \gamma_2 \leq k$,
\[ \frac{1}{m} \sum_{1 \leq j \leq m} Y_{i_{j, \gamma_1}} Y_{i_{j, \gamma_2}} \to E_Q[ E_Q[\tilde Y(d) | \tilde X]^2]~. \]
\eqref{eq:tau2} then follows by noting
\[ \frac{1}{m} \sum_{1 \leq j \leq m} \frac{1}{\binom{k}{2}} \sum_{i < i' \in \lambda_j} Y_i(d) Y_{i'}(d) = \frac{1}{\binom{k}{2}} \sum_{1 \leq \gamma_1 < \gamma_2 \le k} \frac{1}{m} \sum_{1 \leq j \leq m} Y_{i_{j, \gamma_1}} Y_{i_{j, \gamma_2}}~. \]
The result in \eqref{eq:tau3} can be established similarly. Therefore, noting $n = m k$, we get
\begin{equation} \label{eq:tau-limit}
E[\hat \tau_n^2] \to \frac{1}{\ell} E_Q[\var_Q[\tilde Y(1) | \tilde X]] + \frac{1}{k - \ell} E_Q[\var_Q[\tilde Y(0) | \tilde X]] + E_Q[E_Q[\tilde Y(1) - \tilde Y(0) | \tilde X]^2]~.
\end{equation}
Next, recall from the proof of Theorem \ref{thm:main} that
\[E[\hat \kappa_n] = \frac{2}{m} \sum_{1 \leq j \leq \frac{m}{2}} (\bar Y_{2j - 1, n}(1) - \bar Y_{2j - 1, n}(0)) (\bar Y_{2j, n}(1) - \bar Y_{2j, n}(0))~.\]
It follows from Assumption \ref{ass:super}(b) applied across strata and similar arguments to those used to prove \eqref{eq:tau2}--\eqref{eq:tau3} that
\[ E[\hat \kappa_n] \to E_Q[E_Q[\tilde Y(1) - \tilde Y(0) | \tilde X]^2]~. \]
As a result,
\[ n E[\hat V_n] = k (E[\hat \tau_n^2] - E[\hat \kappa_n]) \to E_Q\left[\frac{\var_Q[\tilde{Y}(1)|\tilde{X}]}{\eta} + \frac{\var_Q[\tilde{Y}(0)|\tilde{X}]}{1- \eta}\right]~. \]

\noindent (b) For $n E[\hat{V}_n^{\rm IM}]$, first observe that because $\hat \Delta_n = \frac{1}{m} \sum_{1 \leq j \leq m} \hat \Delta_{j, n}$,
\begin{equation} \label{eq:im-tau}
\begin{split}
n \hat{V}_n^{\rm IM} & = k \frac{m}{m - 1} (m - 1) \hat{V}_n^{\rm IM} \\
& = k \frac{m}{m - 1} \bigg ( \frac{1}{m}\sum_{j=1}^m \hat{\Delta}_{j,n}^2-\hat{\Delta}_n^2 \bigg ) \\
& = k \hat \tau_n^2 - k \frac{m}{m - 1} \frac{2}{m^2} \sum_{1 \leq j < j' \leq m} \hat \Delta_{j, n} \hat \Delta_{j', n}~.
\end{split}
\end{equation}
It therefore suffices to derive the limit of
\[ E \bigg [ \frac{2}{m^2} \sum_{1 \leq j < j' \leq m} \hat \Delta_{j, n} \hat \Delta_{j', n} \bigg ] \]
and combine with the limit of $E[\hat \tau_n^2]$ in \eqref{eq:tau-limit}. To that end, note
\begin{align}
\nonumber & E \bigg [ \frac{2}{m^2} \sum_{1 \leq j < j' \leq m} \hat \Delta_{j, n} \hat \Delta_{j', n} \bigg ] \\
\nonumber & = \frac{2}{m^2} \sum_{1 \leq j < j' \leq m} \Delta_{j, n} \Delta_{j', n} \\
\nonumber & = \bigg ( \frac{1}{m} \sum_{1 \leq j \leq m} \Delta_{j, n} \bigg )^2 - \frac{1}{m^2} \sum_{1 \leq j \leq n} \Delta_{j, n}^2 \\
\nonumber & = \Delta_n^2 - \frac{1}{m^2} \sum_{1 \leq j \leq n} \frac{1}{k^2} \bigg ( \sum_{i \in \lambda_j} (Y_i(1) - Y_i(0))^2 + 2 \sum_{i < i' \in \lambda_j} (Y_i(1) - Y_i(0)) (Y_{i'}(1) - Y_{i'}(0)) \bigg ) \\
\label{eq:im-cross} & \to E_Q[\tilde Y(1) - \tilde Y(0)]^2~,
\end{align}
where the first equality follows because $\hat \Delta_{j, n}$ and $\hat \Delta_{j', n}$ are independent and $E[\hat \Delta_{j, n}] = \Delta_{j, n}$, the second follows by completing the square, the third follows by expanding the square, and the convergence follows because $\Delta_n \to E_Q[\tilde Y(1) - \tilde Y(0)]$, $\frac{1}{m} \to 0$, and
\begin{align*}
\frac{1}{n} \sum_{1 \leq i \leq n} (Y_i(1) - Y_i(0))^2 & \to E_Q[(\tilde Y(1) - \tilde Y(0))^2] \\
\frac{1}{m} \sum_{1 \leq j \leq n} \frac{1}{\binom{k}{2}} \sum_{i < i' \in \lambda_j} (Y_i(1) - Y_i(0)) (Y_{i'}(1) - Y_{i'}(0)) & \to E_Q[E_Q[\tilde Y(1) - \tilde Y(0) | \tilde X]^2]  
\end{align*}
by similar arguments to those used to prove \eqref{eq:tau1}--\eqref{eq:tau2}. Combining \eqref{eq:tau-limit}, \eqref{eq:im-tau}, and \eqref{eq:im-cross}, noting $\frac{m}{m - 1} \to 1$, we have
\begin{align*}
n \hat E[V_n^{\rm IM}] & \to \frac{1}{\eta} E_Q[\var_Q[\tilde Y(1) | \tilde X]] + \frac{1}{1 - \eta} E_Q[\var_Q[\tilde Y(0) | \tilde X]] \\
& \hspace{3em} + k E_Q[E_Q[\tilde Y(1) - \tilde Y(0) | \tilde X]^2] - k E_Q[\tilde Y(1) - \tilde Y(0)]^2 \\
& = V^{\rm obs} + k \var_Q[E_Q[\tilde Y(1) - \tilde Y(0) | X]]~,
\end{align*}
as desired.

\noindent (c) Finally, we study the limit of $n E[V_n^{\rm F}]$. We first show that the diagonal entries of the projection matrix $H_{Q_2}$ converge to $0$. Let $\Sigma_{X, n}= \frac{1}{m} \sum_{1 \le j \le m}(\bar{X}_{j,n}-\mu_{X,n})(\bar{X}_{j,n}-\mu_{X,n})'$, and we have
\begingroup
\allowdisplaybreaks
\begin{align*}
H_{Q_2}&=Q_2(Q_2'Q_2)^{-1}Q_2' \\
&=
\begin{pmatrix}
     & (\bar X_{1, n} - \mu_{X, n})' \\
    \iota_m & \vdots \\
     & (\bar X_{m, n} - \mu_{X, n})'
\end{pmatrix}
\begin{pmatrix}
m & 0_p' \\
0_p & m \Sigma_{X, n}
\end{pmatrix}^{-1}
\begin{pmatrix}
 & \iota_m' & \\
\bar X_{1, n} - \mu_{X, n} & \cdots & \bar X_{m, n} - \mu_{X, n}
\end{pmatrix}
\\
&=
\frac{1}{m}\iota_m\iota_m'
+
\frac{1}{m} \begin{pmatrix}
(\bar X_{1, n} - \mu_{X, n})' \\
\vdots \\
(\bar X_{m, n} - \mu_{X, n})'
\end{pmatrix}
\Sigma_{X, n}^{-1}
\begin{pmatrix}
\bar X_{1, n} - \mu_{X, n} & \cdots & \bar X_{m, n} - \mu_{X, n}
\end{pmatrix}~.
\end{align*}
\endgroup
Therefore, denoting the operator norm by $\|\cdot\|_{\rm op}$, we have that for $1 \le j \le m$, the $j$-th diagonal entry of $H_{Q_2}$ may be uniformly bounded as
\begin{align*}
\max_{1 \leq j \leq m} [H_{Q_2}]_{jj}& = \max_{1 \leq j \leq m} \left ( \frac{1}{m}+\frac{1}{m}(\bar X_{j, n} - \mu_{X, n})' \Sigma_{X, n}^{-1}(\bar X_{j, n} - \mu_{X, n}) \right )
\\
&\le
\frac{1}{m} + \frac{1}{m}\left\lVert \Sigma_{X, n}^{-1} \right\rVert_{\rm op} \max_{1 \leq j \leq m} \lVert \bar X_{j, n} - \mu_{X, n} \rVert^2
\\
&\lesssim
\frac{1}{m}+ k \left\lVert \Sigma_{X, n}^{-1}\right\rVert_{\rm op} \frac{1}{n}\max_{1 \le i \le n} \lVert X_i \rVert^2 \to 0,
\end{align*}
where the first inequality follows from the definition of the operator norm, the second inequality follows from repeated application of the inequality that $(a - b)^2 \leq 2(a^2 + b^2)$, and the convergence follows from Assumption \ref{ass:X}(a) and that $\Sigma_{X, n} \to \var_Q[\tilde X]$, which further follows from Assumption \ref{ass:X}(b)--(c) by similar arguments to those used to prove \eqref{eq:tau1}--\eqref{eq:tau2}.

Define
\[Q_{2, X} = \begin{pmatrix}
     (\bar X_{1, n} - \mu_{X, n})' \\
     \vdots \\
     (\bar X_{m, n} - \mu_{X, n})'
\end{pmatrix}~. \]
Let $H_{\iota_m}$ denote the projection matrix for $\iota_m$ and $H_{Q_{2, X}}$ denote the projection matrix for $Q_{2, X}$. Note that by construction $\iota_m' Q_{2, X} = 0$. It follows from Lemma \ref{lem:quadratic} that
\begin{align}
\nonumber n E[\hat{V}_n^{\rm F}] &= \frac{k}{m} E[\hat{\delta}_n'(\diag(I-H_{Q_2}))^{-1/2}(I-H_{Q_2})(\diag(I-H_{Q_2}))^{-1/2}\hat{\delta}_n]
\\
\label{eq:f1} & = \frac{k}{m} \delta_n' (\diag(I-H_{Q_2}))^{-1/2}(I-H_{Q_2})(\diag(I-H_{Q_2}))^{-1/2} \delta_n \\
\label{eq:f2} & \hspace{3em} + \frac{k}{m} \tr((I-H_{Q_2}) \var[\diag(I-H_{Q_2}))^{-1/2} \hat \delta_n])
\end{align}
Define $\epsilon_n = (\diag(I-H_{Q_2}))^{-1/2} \delta_n - \delta_n$. Note \eqref{eq:f1} equals
\begin{equation} \label{eq:f1-expanded}
\frac{k}{m} \delta_n' (I - H_{Q_2}) \delta_n + \frac{k}{m} \epsilon_n' (I - H_{Q_2}) \delta_n + \frac{k}{m} \delta_n' (I - H_{Q_2}) \epsilon_n + \frac{k}{m} \epsilon_n' (I - H_{Q_2}) \epsilon_n~.    
\end{equation}
Note $\frac{1}{m} \|\delta_n\|^2$ converges to a constant by the derivations in the proof for the limit of $n E[\hat V_n^{\rm IM}]$. Because $\max_{1 \leq j \leq m} [H_{Q_2}]_{jj} \to 0$ as $n \to \infty$,
\[ \frac{1}{m} \|\epsilon_n\|^2 \leq \frac{1}{m} \|\delta_n\|^2 \max_{1 \leq j \leq m} \left ( \frac{1}{\sqrt{1 - [H_{Q_2}]_{jj}}} - 1 \right )^2 \to 0~. \]
Because $I - H_{Q_2}$ is a projection matrix, all of its eigenvalues are either 0 or 1, and hence its operator norm is 1. Therefore,
\[ \frac{k}{m} \epsilon_n' (I - H_{Q_2}) \delta_n \leq k \left ( \frac{1}{m}\|\epsilon_n\|^2 \right )^{1/2} \|I - H_{Q_2}\|_{\rm op} \left ( \frac{1}{m} \|\delta_n\| \right )^{1/2} \to 0~. \]
Similarly, the last two terms in \eqref{eq:f1-expanded} also converge to 0 as $n \to \infty$. Recalling $\iota_m' Q_{2, X} = 0$, we get that \eqref{eq:f1} equals
\begin{equation} \label{eq:muAmu}
\frac{k}{m} \delta_n' (I - H_{Q_2}) \delta_n = \frac{k}{m} \delta_n' (I-H_{\iota_m}-H_{Q_{2, X}}) \delta_n = \frac{k}{m} \delta_n' (I-H_{\iota_m}) \delta_n - \frac{k}{m} \delta_n' Q_{2, X} \Sigma_{X, n}^{-1} \frac{1}{m} Q_{2, X}' \delta_n~.
\end{equation}
Next, because $\hat \Delta_{j, n}$ are independent across $1 \leq j \leq m$, note \eqref{eq:f2} equals
\begin{align}
\nonumber & \frac{k}{m} \tr \left ( (I-H_{Q_2}) \diag \bigg ( \frac{\var[\hat \Delta_{j, n}]}{1 - [H_{Q_2}]_{jj}}: 1 \leq j \leq m \bigg ) \right ) \\
\nonumber & = \frac{k}{m} \sum_{1 \leq j \leq m} (1 - [H_{Q_2}]_{jj}) \frac{\var[\hat \Delta_{j, n}]}{1 - [H_{Q_2}]_{jj}} \\
\label{eq:Avar} & = \frac{k}{m} \sum_{1 \leq j \leq m} \var[\hat \Delta_{j, n}]~.
\end{align}
Moreover, this relationship holds with $Q_2$ replaced by $\iota_m$. Therefore, the \eqref{eq:muAmu}--\eqref{eq:Avar} imply
\begin{align*}
& \nonumber n E[\hat{V}_n^{\rm F}] \\
& = \frac{k}{m} \delta_n' (I-H_{\iota_m}) \delta_n - \frac{k}{m} \delta_n' Q_{2, X} \Sigma_{X, n}^{-1} \frac{1}{m} Q_{2, X}' \delta_n + \frac{k}{m} \tr((I - H_{\iota_m}) \var[(\diag(I - H_{\iota_m}))^{-1/2} \hat \delta_n]) \\
& = n E[\hat V_n^{\rm IM}] - \frac{k}{m} \delta_n' Q_{2, X} \Sigma_{X, n}^{-1} \frac{1}{m} Q_{2, X}' \delta_n
\end{align*}
Next, note $\frac{1}{m} Q_{2, X}' \delta_n \to \cov_Q[\tilde X,\tilde Y(1)-\tilde Y(0)]$ under Assumption \ref{ass:X}(a)--(c) by similar arguments to those used to prove \eqref{eq:tau1}--\eqref{eq:tau2}. Further note $\frac{m - 1}{m} \to 1$ and recall $\Sigma_{X, n} \to \var_Q[\tilde X]$ and the limit of $nE[\hat V_n^{\rm IM}]$ derived above, and we have
\begin{align*}
n E[\hat{V}_n^{\rm F}] &\to
V^{\rm obs}
+
k\var_Q[ E_Q[ \tilde{Y}(1)-\tilde{Y}(0)|\tilde{X} ] ]
-
k\cov_Q[\tilde X,\tilde Y(1)-\tilde Y(0)]\var_Q[\tilde X]^{-1}\cov_Q[\tilde X,\tilde Y(1)-\tilde Y(0)]
\\
&=
V^{\rm obs}
+
k\var_Q[ E_Q[ \tilde{Y}(1)-\tilde{Y}(0)|\tilde{X} ] ]
-
k\var_Q[\text{BLP}_Q(\tilde Y(1)-\tilde Y(0)|1,\tilde X)]
\\
&=
V^{\rm obs}
+
k\var_Q[E_Q[ \tilde{Y}(1)-\tilde{Y}(0)|\tilde{X} ]-\text{BLP}_Q(\tilde Y(1)-\tilde Y(0)|1,\tilde X)]~,
\end{align*}
where the first equality follows from Lemma \ref{lem:blp}(b) and the second follows from Lemma \ref{lem:blp}(d). 

Finally, the statement about the coverages of confidence intervals follow from Theorem \ref{thm:normal} (which applies because of \eqref{eq:pop_moments}--\eqref{eq:pop_nondegen}), \eqref{eq:consistent}, and Slutsky's theorem. \qed

\subsection{Auxiliary Results}

Here we state and prove a standard central limit theorem for design-based inference \citep[see for instance][]{hajek1960limiting,li2017general}.

\begin{theorem}\label{thm:normal}
Suppose that Assumption \ref{ass:dn} holds and that
\begin{equation} \label{eq:highlevel}
\frac{\max_{1 \leq j \leq m} \max_{i \in \lambda_j} (R_i - \bar R_j)^2}{\sum_{1 \leq j \leq m} \sum_{i \in \lambda_j} (R_i - \bar R_j)^2} \rightarrow 0
\end{equation}
as $n \rightarrow \infty$, where $R_i := \frac{Y_i(1)}{\eta} + \frac{Y_i(0)}{1 - \eta}$ and $\bar R_j := \frac{1}{k} \sum_{i \in \lambda_j} R_i$.  Then, $$\frac{\hat \Delta_n - \Delta_n}{\var[\hat \Delta_n]^{1/2}} \stackrel{d}{\rightarrow} N(0,1)~.$$
\end{theorem}

\begin{proof}
Re-writing $\hat{\Delta}_n$ and $\Delta_n$ over the blocks $\{\lambda_j: 1 \le j \le n\}$, we obtain
\[\hat{\Delta}_n = \frac{1}{n}\left(\sum_{1 \leq j \leq m}\sum_{i \in \lambda_j}\left(\frac{Y_i(1)}{\eta} + \frac{Y_i(0)}{1 - \eta}\right)D_i - \sum_{1 \leq j \leq m}\sum_{i \in \lambda_j}\frac{Y_i(0)}{1 - \eta}\right)~,\]
and 
\[\Delta_n = \frac{1}{n}\left(\sum_{1 \leq j \leq m}\sum_{i \in \lambda_j}\left(\bar{R}_j(1) + \bar{R}_j(0)\right)D_i - \sum_{1 \leq j \leq m}\sum_{i \in \lambda_j}\frac{Y_i(0)}{1 - \eta}\right)~,\]
where $\bar{R}_j(1) = \frac{1}{k}\sum_{i \in \lambda_j}\frac{Y_i(1)}{\eta}$ and similarly for $\bar{R}_j(0)$. Putting both together,
\[\hat{\Delta}_n - \Delta_n = \frac{1}{n}\sum_{1 \leq j \leq m}\sum_{i \in \lambda_j}(R_i - \bar{R}_j)D_i~,\]
where $R_i = \frac{Y_i(1)}{\eta} + \frac{Y_i(0)}{1 - \eta}$ and $\bar{R}_j = \bar{R}_j(1) + \bar{R}_j(0)$. Hence we obtain that 
\[\sqrt{m}(\hat{\Delta}_n - \Delta_n) = \sum_{1 \leq j \leq m}A_j~,\]
where $A_j = \frac{1}{\sqrt{m}}\frac{\eta}{\ell}\sum_{i \in \lambda_j}(R_i - \bar{R}_j)D_i$.
Note that $\frac{1}{\ell}\sum_{i \in \lambda_j}(R_i - \bar{R}_j)D_i$ is the sample mean when sampling a subset of size $\ell$ without replacement from the finite population $\{(R_i - \bar{R}_j): i \in \lambda_j\}$. Accordingly, from the properties of the sample mean when sampling from a finite population (Theorems 2.1 and 2.2 in \cite{cochran1977sampling}), we obtain immediately that 
\[E[A_j] = 0~,\]
\[\var[A_j] = \frac{1}{m}\eta^2\left(\frac{1}{\ell} - \frac{1}{k}\right)\left(\frac{1}{k-1}\right)\sum_{i \in \lambda_j}(R_i - \bar{R}_j)^2~.\]
Moreover, by the definition of $A_j$ and Jensen's inequality, 
\[A_j^2 \lesssim \frac{1}{m}\max_{i \in \lambda_j}(R_i - \bar{R}_j)^2~.\]
To establish asymptotic normality, we verify the Lindeberg condition:
\[ \sum_{1 \leq j \leq m} E \left [ \frac{A_j^2}{s_n^2} I \left\{\frac{A_j^2}{s_n^2} > \epsilon^2\right\} \right ] \to 0~,\]
where $s_n^2 = \sum_{1 \leq j \leq m}\var[A_j]$. Note that by our previous calculations 
\[ s_n^2 \propto \frac{1}{m}\sum_{1 \leq j \leq m} \sum_{i \in \lambda_j} (R_i - \bar R_j)^2~,\]
hence 
\[\frac{A_j^2}{s_n^2} \lesssim \frac{\max_{i \in \lambda_j}(R_i - \bar{R}_j)^2}{\sum_{1 \leq j \leq m}\sum_{i \in \lambda_j}(R_i - \bar{R}_j)^2}~.\]
Therefore, if \eqref{eq:highlevel} holds, then for each fixed $\epsilon > 0$, for $m$ large enough,
\[ \max_{1 \leq j \leq m} \frac{A_j^2}{s_n^2} \leq \epsilon^2~, \]
so
\[ \sum_{1 \leq j \leq m} \frac{A_j^2}{s_n^2} I \left\{\frac{A_j^2}{s_n^2} > \epsilon^2\right\} \leq \sum_{1 \leq j \leq m} \frac{A_j^2}{s_n^2} I \left\{\max_{1 \leq j \leq m} \frac{A_j^2}{s_n^2} > \epsilon^2\right\} = 0~. \]
As a result, the Lindeberg condition holds, and the conclusion follows.
\end{proof}

\begin{lemma} \label{lem:ratio}
Let $X_n, n \geq 1$ be a sequence of random variables such that $E[X_n] = 0$, $\var[X_n] = \sigma_n^2$ and $X_n / \sigma_n \stackrel{d}{\to} N(0, 1)$. Further suppose $\liminf_{n \to \infty} \sigma_n^2 > 0$, $\limsup_{n \to \infty} \sigma_n^2 < \infty$, and there exists a sequence of random variables $\hat \sigma_n^2, n \geq 1$ such that $\hat \sigma_n^2 - E[\hat \sigma_n^2] \stackrel{P}{\to} 0$. Then,
\begin{align*}
\liminf_{n \to \infty} P \left \{ \frac{X_n}{\hat \sigma_n} \leq x \right \} \geq  \Phi(x) & \text{ for } x > 0 \\
\limsup_{n \to \infty} P \left \{ \frac{X_n}{\hat \sigma_n} \leq x \right \} \leq  \Phi(x) & \text{ for } x < 0~,
\end{align*}
where $\Phi(x) = P \{N(0, 1) \leq x\}$, if and only if $\liminf_{n \to \infty} (E[\hat \sigma_n^2] - \sigma_n^2) \geq 0$.
\end{lemma}

\begin{proof}
We will only prove the inequality when $x \geq 0$ and the other follows similarly. Because $\liminf_{n \to \infty} \sigma_n^2 > 0$, $\sigma_n^2 \ge c > 0$ for some $c > 0$, at least for $n$ large enough. Since further $\liminf_{n \to \infty} (E[\hat \sigma_n^2] - \sigma_n^2) \geq 0$, $E[\hat \sigma_n^2] \geq \frac{3}{4} c$ for $n$ large enough. Because further $\hat \sigma_n^2 - E[\hat \sigma_n^2] \stackrel{P}{\to} 0$, we have
\[ P \{\hat \sigma_n^2 > c / 2\} \to 1~. \]
Therefore, in what follows, all statements should be understood as conditioning on this event. In that case,
\begin{equation} \label{eq:x/v}
\frac{X_n}{\hat \sigma_n} = \frac{X_n}{\sigma_n} \frac{\sigma_n}{\hat \sigma_n}~.    
\end{equation}
Next, we claim that for each $\epsilon > 0$,
\begin{equation} \label{eq:v/v}
P \left \{ \frac{\sigma_n}{\hat \sigma_n} > 1 + \epsilon \right \} \to 0~.    
\end{equation}
Suppose for a moment that \eqref{eq:v/v} holds. Then, because $X_n / \sigma_n \stackrel{d}{\to} N(0, 1)$, $\sigma_n > 0$, and $\hat \sigma_n > 0$, \eqref{eq:x/v} implies that for each $x \geq 0$ and $\epsilon > 0$,
\[ P \left \{ \frac{X_n}{\hat \sigma_n} > x (1 + \epsilon) \right \} \leq P \left \{ \frac{X_n}{\sigma_n} > x \right \} + P \left \{ \frac{\sigma_n}{\hat \sigma_n} > 1 + \epsilon \right \} \to 1 - \Phi(x) \]
as $n \to \infty$. In other words, for $x > 0$ and $\epsilon > 0$,
\[ \limsup_{n \to \infty} P \left \{ \frac{X_n}{\hat \sigma_n} > x \right \} \leq 1 - \Phi(x / (1 + \epsilon))~. \]
The conclusion of the theorem follows by letting $\epsilon \to 0$ and noting $\Phi$ is continuous everywhere. It therefore suffices to prove \eqref{eq:v/v}. Because \eqref{eq:v/v} holds if and only if
\begin{equation} \label{eq:indicatorP}
I \left \{ \frac{\sigma_n}{\hat \sigma_n} > 1 + \epsilon \right \} \stackrel{P}{\to} 0~,    
\end{equation}
it suffices to prove that for each subsequence $\{n(k)\}$, there exists a further subsequence $\{n(k(\ell))\}$ along which the convergence in \eqref{eq:indicatorP} holds with probability one. Fix an arbitrary subsequence $\{n(k)\}$. Because $\liminf_{n \to \infty} \sigma_n^2 > 0$, $\liminf_{n \to \infty} (E[\hat \sigma_n^2] - \sigma_n^2) \geq 0$, and $\hat \sigma_n^2 - E[\hat \sigma_n^2] \stackrel{P}{\to} 0$, there exists a constant $c > 0$ and a further subsequence $\{n(k(\ell))\}$ along which $\sigma_{n(k(\ell))}^2 \geq c > 0$, $E[\hat \sigma_{n(k(\ell))}^2] \geq c / 2 > 0$, and $\hat \sigma_{n(k(\ell))}^2 - E[\hat \sigma_{n(k(\ell))}^2] \to 0$ with probability one. Furthermore,
\[ \limsup_{\ell \to \infty} \frac{\sigma_{n(k(\ell))}^2}{E[\hat \sigma_{n(k(\ell))}^2]} \leq 1~. \]
Along this subsequence,
\[ \limsup_{\ell \to \infty} \frac{\sigma_{n(k(\ell))}^2}{\hat \sigma_{n(k(\ell))}^2} = \limsup_{\ell \to \infty} \frac{\sigma_{n(k(\ell))}^2}{E[\hat \sigma_{n(k(\ell))}^2]} \frac{E[\hat \sigma_{n(k(\ell))}^2]}{\hat \sigma_{n(k(\ell))}^2} \leq 1~, \]
so \eqref{eq:indicatorP} holds with probability one, and hence \eqref{eq:v/v} follows.

Next, we show that $\liminf_{n \rightarrow \infty} (E[\hat{\sigma}_n^2] - \sigma_n^2) \ge 0$ is necessary for the result to hold. We present the proof for $x > 0$ while symmetric arguments work for $x < 0$. To that end, first suppose $\liminf_{n \rightarrow \infty} (E[\hat{\sigma}_n^2] - \sigma_n^2) = - c < 0$ and $\liminf_{n \to \infty} E[\hat \sigma_n^2] > 0$. Then, there exists a subsequence $\{n(k)\}$ along which $E[\hat \sigma_{n(k)}^2] < \sigma_{n(k)}^2 - c / 2$ and $E[\hat \sigma_{n(k)}^2] \ge c/2$. Because by assumption $\limsup_{n \to \infty} \sigma_n^2 < \infty$, there exists $M > 0$ such that $\sigma_{n(k)}^2 \leq M$ for all $k$. Define $\epsilon = \frac{M}{M - c / 2} - 1 > 0$. By arguing similarly as in \eqref{eq:indicatorP}, we can show
\begin{equation} \label{eq:indicator2}
P \left \{ \frac{\sigma_{n(k)}}{\hat \sigma_{n(k)}} \geq 1 + \epsilon \right \} \to 1~.
\end{equation}
If instead $\liminf_{n \to \infty} E[\hat \sigma_n^2] = 0$, then by passing to a subsequence along which $E[\hat \sigma_{n(k)}^2] \to 0$, we can construct a further subsequence along which
\[\frac{\sigma_{n(k(\ell))}}{\hat \sigma_{n(k(\ell))}}\rightarrow \infty~,\]
with probability one, and thus \eqref{eq:indicator2} still holds. Then,
\begin{align*}
P \left \{ \frac{X_{n(k)}}{\hat \sigma_{n(k)}} > x \right \} = P \left \{ \frac{X_{n(k)}}{\sigma_{n(k)}} \frac{\sigma_{n(k)}}{\hat \sigma_{n(k)}} > x \right \} & \geq P \left \{ \Big \{ \frac{X_{n(k)}}{\sigma_{n(k)}}  > \frac{x}{1 + \epsilon} \Big \} \cap \Big \{ \frac{\sigma_{n(k)}}{\hat \sigma_{n(k)}} \geq 1 + \epsilon \Big \} \right \} \\
& = P \left \{ \frac{X_{n(k)}}{\sigma_{n(k)}}  > \frac{x}{1 + \epsilon}\right \} - P \left \{ \Big \{ \frac{X_{n(k)}}{\sigma_{n(k)}}  > \frac{x}{1 + \epsilon} \Big \} \cap \Big \{ \frac{\sigma_{n(k)}}{\hat \sigma_{n(k)}} < 1 + \epsilon \Big \} \right \} \\
& \to 1 - \Phi(x / (1 + \epsilon)) > 1 - \Phi(x)~.
\end{align*}
Therefore,
\[ \liminf_{n \to \infty} P \left \{ \frac{X_n}{\hat \sigma_n} \leq x \right \} = 1 - \limsup _{n \to \infty} P \left \{ \frac{X_n}{\hat \sigma_n} > x \right \} \leq \Phi(x / (1 + \epsilon)) < \Phi(x)~, \]
and the desired result follows.
\end{proof}

\begin{lemma} \label{lem:quadratic}
Suppose $E[\hat \delta_n] = \delta_n$ and $A$ is deterministic. Then, $E[\hat \delta_n' A \hat \delta_n] = \delta_n' A \delta_n + \tr(A \var[\hat \delta_n])$.   
\end{lemma}

\begin{proof}
Note
\begin{align*}
E[\hat \delta_n' A \hat \delta_n] & = E[(\delta_n + \hat \delta_n - \delta_n)' A ( \delta_n + \hat \delta_n - \delta_n)] \\
& = \delta_n' A \delta_n + E[(\hat \delta_n - \delta_n)' A (\hat \delta_n - \delta_n)] \\
& = \delta_n' A \delta_n + \tr( E[(\hat \delta_n - \delta_n)' A (\hat \delta_n - \delta_n)]) \\
& = \delta_n' A \delta_n + E[\tr((\hat \delta_n - \delta_n)' A (\hat \delta_n - \delta_n))] \\
& = \delta_n' A \delta_n + E[\tr(A (\hat \delta_n - \delta_n) (\hat \delta_n - \delta_n)')] \\
& = \delta_n' A \delta_n + \tr(E[A (\hat \delta_n - \delta_n) (\hat \delta_n - \delta_n)']) \\
& = \delta_n' A \delta_n + \tr(A \var[\hat \delta_n])~,
\end{align*}
where the second equality follows because $E[\hat \delta_n - \delta_n] = 0$, the third follows because $E[(\hat \delta_n - \delta_n)' A (\hat \delta_n - \delta_n)]$ is a scalar, the fourth and sixth inequalities follow because the expectation operator is linear, and the fifth follows from $\tr(AB) = \tr(BA)$ as long as the matrices are commutative.
\end{proof}

\begin{lemma} \label{lem:blp}
Let $(Y, X)$ be random vectors where $Y$ lives in $\mathbf R$ and $X$ lives in $\mathbf R^k$. Then,
\begin{enumerate}[\rm (a)]
    \item $E[\blp(Y | 1, X)] = E[Y]$.
    % \item $\blp(Y | 1, X) = \blp(Y | 1, X - E[X])$.
    \item $(X - E[X])' \var[X]^{-1} \cov[X, Y] = \blp(Y | 1, X) - E[\blp(Y | 1, X)]$.
    \item $E[(Y - E[Y] - (\blp(Y | 1, X) - E[\blp(Y | 1, X)])) (\blp(Y | 1, X) - E[\blp(Y | 1, X)])] = 0$ and therefore $\cov[Y, \blp(Y|1, X)] = \var[\blp(Y|1, X)]$.
    \item $E([E[Y | X] - E[Y] - (\blp(Y | 1, X) - E[\blp(Y | 1, X)])) (\blp(Y | 1, X) - E[\blp(Y | 1, X)])] = 0$ and therefore $\cov[E[Y | X], \blp(Y| 1, X)] = \var[\blp(Y|1, X)]$ and $\var[E[Y | X]] = \var[E[Y | X] - \blp(Y|1, X)] + \var[\blp(Y|1, X)]$.
\end{enumerate}
\end{lemma}

\begin{proof}
(a) follows because $E[(Y - \blp(Y | 1, X)) \cdot 1] = 0$. (b) follows because if $\blp(Y | 1, X) = \beta_0 + X' \beta_1$ then $\beta_1 = \var[X]^{-1} \cov[X, Y]$. (c) follows because recalling $E[\blp(Y | 1, X)] = E[Y]$, we have
\begin{align*}
& E[(Y - E[Y] - (\blp(Y | 1, X) - E[\blp(Y | 1, X)])) (\blp(Y | 1, X) - E[\blp(Y | 1, X)])] \\
& = E[(Y - \blp(Y | 1, X)) (\blp(Y | 1, X) - E[Y])] \\
& = E[(Y - \blp(Y | 1, X)) \blp(Y | 1, X)] + E[Y - \blp(Y | 1, X)] E[Y] \\
& = 0~.
\end{align*}
(d) follows similarly.
\end{proof}

\newpage
\bibliography{finpop}

\begin{thebibliography}{28}
\expandafter\ifx\csname natexlab\endcsname\relax\def\natexlab#1{#1}\fi
\expandafter\ifx\csname url\endcsname\relax
  \def\url#1{\texttt{#1}}\fi
\expandafter\ifx\csname urlprefix\endcsname\relax\def\urlprefix{URL }\fi
\providecommand{\eprint}[2][]{\url{#2}}

\bibitem[{Abadie and Imbens(2008)}]{abadie2008estimation}
\textsc{Abadie, A.} and \textsc{Imbens, G.~W.} (2008).
\newblock Estimation of the {Conditional} {Variance} in {Paired} {Experiments}.
\newblock \textit{Annales d'Économie et de Statistique} 175--187.

\bibitem[{Aronow et~al.(2014)Aronow, Green and Lee}]{aronow2014sharp}
\textsc{Aronow, P.~M.}, \textsc{Green, D.~P.} and \textsc{Lee, D. K.~K.} (2014).
\newblock Sharp {Bounds} on the {Variance} in {Randomized} {Experiments}.
\newblock \textit{The Annals of Statistics}, \textbf{42} 850--871.

\bibitem[{Bai(2022)}]{bai2022optimality}
\textsc{Bai, Y.} (2022).
\newblock Optimality of {Matched}-{Pair} {Designs} in {Randomized} {Controlled} {Trials}.
\newblock \textit{American Economic Review}, \textbf{112} 3911--3940.

\bibitem[{Bai et~al.(2024{\natexlab{a}})Bai, Guo, Shaikh and Tabord-Meehan}]{bai2024inference-b}
\textsc{Bai, Y.}, \textsc{Guo, H.}, \textsc{Shaikh, A.~M.} and \textsc{Tabord-Meehan, M.} (2024{\natexlab{a}}).
\newblock Inference in {Experiments} with {Matched} {Pairs} and {Imperfect} {Compliance}.
\newblock \textit{Journal of Business \& Economic Statistics}, \textbf{0} 1--22.

\bibitem[{Bai et~al.(2024{\natexlab{b}})Bai, Jiang, Romano, Shaikh and Zhang}]{bai2024covariate}
\textsc{Bai, Y.}, \textsc{Jiang, L.}, \textsc{Romano, J.~P.}, \textsc{Shaikh, A.~M.} and \textsc{Zhang, Y.} (2024{\natexlab{b}}).
\newblock Covariate adjustment in experiments with matched pairs.
\newblock \textit{Journal of Econometrics}, \textbf{241} 105740.

\bibitem[{Bai et~al.(2024{\natexlab{c}})Bai, Liu, Shaikh and Tabord-Meehan}]{bai2024inferencecluster}
\textsc{Bai, Y.}, \textsc{Liu, J.}, \textsc{Shaikh, A.~M.} and \textsc{Tabord-Meehan, M.} (2024{\natexlab{c}}).
\newblock Inference in cluster randomized trials with matched pairs.
\newblock \textit{Journal of Econometrics}, \textbf{245} 105873.

\bibitem[{Bai et~al.(2025)Bai, Liu, Shaikh and Tabord-Meehan}]{bai2025efficiency}
\textsc{Bai, Y.}, \textsc{Liu, J.}, \textsc{Shaikh, A.~M.} and \textsc{Tabord-Meehan, M.} (2025).
\newblock On the {Efficiency} of {Finely} {Stratified} {Experiments}.
\newblock ArXiv:2307.15181 [econ], \urlprefix\url{http://arxiv.org/abs/2307.15181}.

\bibitem[{Bai et~al.(2024{\natexlab{d}})Bai, Liu and Tabord-Meehan}]{bai2024inference}
\textsc{Bai, Y.}, \textsc{Liu, J.} and \textsc{Tabord-Meehan, M.} (2024{\natexlab{d}}).
\newblock Inference for {Matched} {Tuples} and {Fully} {Blocked} {Factorial} {Designs}.
\newblock \textit{Quantitative Economics}, \textbf{15} 279--330.

\bibitem[{Bai et~al.(2022)Bai, Romano and Shaikh}]{bai2022inference}
\textsc{Bai, Y.}, \textsc{Romano, J.~P.} and \textsc{Shaikh, A.~M.} (2022).
\newblock Inference in {Experiments} {With} {Matched} {Pairs}.
\newblock \textit{Journal of the American Statistical Association}, \textbf{117} 1726--1737.

\bibitem[{Cochran(1977)}]{cochran1977sampling}
\textsc{Cochran, W.~G.} (1977).
\newblock Sampling techniques.
\newblock \textit{Johan Wiley \& Sons Inc}.

\bibitem[{Cytrynbaum(2024{\natexlab{a}})}]{cytrynbaum2024covariate}
\textsc{Cytrynbaum, M.} (2024{\natexlab{a}}).
\newblock Covariate adjustment in stratified experiments.
\newblock \textit{Quantitative Economics}, \textbf{15} 971--998.

\bibitem[{Cytrynbaum(2024{\natexlab{b}})}]{cytrynbaum2021optimal}
\textsc{Cytrynbaum, M.} (2024{\natexlab{b}}).
\newblock Optimal stratification of survey experiments.
\newblock \textit{arXiv preprint arXiv:2111.08157}.

\bibitem[{De~Chaisemartin and Ramirez-Cuellar(2024)}]{de2024level}
\textsc{De~Chaisemartin, C.} and \textsc{Ramirez-Cuellar, J.} (2024).
\newblock At what level should one cluster standard errors in paired and small-strata experiments?
\newblock \textit{American Economic Journal: Applied Economics}, \textbf{16} 193--212.

\bibitem[{Ding(2017)}]{ding2017paradox}
\textsc{Ding, P.} (2017).
\newblock A {Paradox} from {Randomization}-{Based} {Causal} {Inference}.
\newblock \textit{Statistical Science}, \textbf{32} 331--345.

\bibitem[{Fisher(1935)}]{fisher1935design}
\textsc{Fisher, R.~A.} (1935).
\newblock The design of experiments.
\newblock \textit{The design of experiments.}

\bibitem[{Fogarty(2018{\natexlab{a}})}]{fogarty2018mitigating}
\textsc{Fogarty, C.~B.} (2018{\natexlab{a}}).
\newblock On mitigating the analytical limitations of finely stratified experiments.
\newblock \textit{Journal of the Royal Statistical Society: Series B (Statistical Methodology)}, \textbf{80} 1035--1056.

\bibitem[{Fogarty(2018{\natexlab{b}})}]{fogarty2018regression-assisted}
\textsc{Fogarty, C.~B.} (2018{\natexlab{b}}).
\newblock Regression-assisted inference for the average treatment effect in paired experiments.
\newblock \textit{Biometrika}, \textbf{105} 994--1000.

\bibitem[{Greevy et~al.(2004)Greevy, Lu, Silber and Rosenbaum}]{greevy2004optimal}
\textsc{Greevy, R.}, \textsc{Lu, B.}, \textsc{Silber, J.~H.} and \textsc{Rosenbaum, P.} (2004).
\newblock Optimal multivariate matching before randomization.
\newblock \textit{Biostatistics}, \textbf{5} 263--275.

\bibitem[{H{\'a}jek(1960)}]{hajek1960limiting}
\textsc{H{\'a}jek, J.} (1960).
\newblock Limiting distributions in simple random sampling from a finite population.
\newblock \textit{Publications of the Mathematical Institute of the Hungarian Academy of Sciences}, \textbf{5} 361--374.

\bibitem[{Imai(2008)}]{imai2008variance}
\textsc{Imai, K.} (2008).
\newblock Variance identification and efficiency analysis in randomized experiments under the matched-pair design.
\newblock \textit{Statistics in medicine}, \textbf{27} 4857--4873.

\bibitem[{Imai et~al.(2009)Imai, King and Nall}]{imai2009essential}
\textsc{Imai, K.}, \textsc{King, G.} and \textsc{Nall, C.} (2009).
\newblock The {Essential} {Role} of {Pair} {Matching} in {Cluster}-{Randomized} {Experiments}, with {Application} to the {Mexican} {Universal} {Health} {Insurance} {Evaluation}.
\newblock \textit{Statistical Science}, \textbf{24} 29--53.

\bibitem[{Imbens and Rubin(2015)}]{imbens2015causal}
\textsc{Imbens, G.~W.} and \textsc{Rubin, D.~B.} (2015).
\newblock \textit{Causal inference in statistics, social, and biomedical sciences}.
\newblock Cambridge university press.

\bibitem[{Jiang et~al.(2024)Jiang, Liu, Phillips and Zhang}]{jiang2024bootstrap}
\textsc{Jiang, L.}, \textsc{Liu, X.}, \textsc{Phillips, P.~C.} and \textsc{Zhang, Y.} (2024).
\newblock Bootstrap inference for quantile treatment effects in randomized experiments with matched pairs.
\newblock \textit{Review of Economics and Statistics}, \textbf{106} 542--556.

\bibitem[{Li and Ding(2017)}]{li2017general}
\textsc{Li, X.} and \textsc{Ding, P.} (2017).
\newblock General forms of finite population central limit theorems with applications to causal inference.
\newblock \textit{Journal of the American Statistical Association}, \textbf{112} 1759--1769.

\bibitem[{Liu and Yang(2020)}]{liu2020regression-adjusted}
\textsc{Liu, H.} and \textsc{Yang, Y.} (2020).
\newblock Regression-adjusted average treatment effect estimates in stratified randomized experiments.
\newblock \textit{Biometrika}, \textbf{107} 935--948.

\bibitem[{Pashley and Miratrix(2021)}]{pashley2021insights}
\textsc{Pashley, N.~E.} and \textsc{Miratrix, L.~W.} (2021).
\newblock Insights on variance estimation for blocked and matched pairs designs.
\newblock \textit{Journal of Educational and Behavioral Statistics}, \textbf{46} 271--296.

\bibitem[{Su and Ding(2021)}]{su2021model}
\textsc{Su, F.} and \textsc{Ding, P.} (2021).
\newblock Model-assisted analyses of cluster-randomized experiments.
\newblock \textit{Journal of the Royal Statistical Society Series B: Statistical Methodology}, \textbf{83} 994--1015.

\bibitem[{Zhu et~al.(2024)Zhu, Liu and Yang}]{zhu2024design-based}
\textsc{Zhu, K.}, \textsc{Liu, H.} and \textsc{Yang, Y.} (2024).
\newblock Design-{Based} {Theory} for {Lasso} {Adjustment} in {Randomized} {Block} {Experiments} and {Rerandomized} {Experiments}.
\newblock \textit{Journal of Business \& Economic Statistics}, \textbf{0} 1--12.

\end{thebibliography}
\end{document}